\documentclass{article}

\usepackage[utf8]{inputenc}
\usepackage[margin = 1in]{geometry}

\usepackage[shortlabels]{enumitem}

\usepackage{amsmath}
\usepackage{amsthm}
\usepackage{amsfonts}
\usepackage{bbold}
\usepackage{color}

\usepackage[ruled]{algorithm2e}
\usepackage{hyperref}

\usepackage{graphicx}

\usepackage{natbib}
\usepackage{graphicx}

\usepackage{caption}
\usepackage{subcaption}

\newtheorem{theorem}{Theorem}
\newtheorem{lemma}{Lemma}
\newtheorem{claim}{Claim}

\newtheorem{observation}{Observation}
\newtheorem{remark}{Remark}

\theoremstyle{definition}

\newcommand{\bbE}{\mathbb{E}}
\newcommand{\bbP}{\mathbb{P}}
\newcommand{\bbN}{\mathbb{N}}

\newcommand{\calB}{\mathcal{B}}

\newcommand{\bigO}{\mathcal{O}}
\newcommand{\calN}{\mathcal{N}}
\newcommand{\calG}{\mathcal{G}}

\newcommand{\bbOne}{\mathbb{1}}
\newcommand{\pa}[1]{\left( #1 \right)}

\newcommand{\var}{\textnormal{Var}}

\newcommand{\tendsto}[2]{\underset{#1 \rightarrow #2}{\longrightarrow}}

\newcommand{\smallO}[2]{\underset{#1 \rightarrow #2}{o}}

\newcommand*{\medcap}{\mathbin{\scalebox{1.5}{\ensuremath{\cap}}}}

\newcommand{\prob}[5]{\bbP\pa{ B_{#1}\pa{#2} #3 B_{#4}\pa{#5} }} 

\newcommand{\cnt}{\cn'}
\newcommand{\cntt}{\cn''}
\SetAlgorithmName{Protocol}{protocol}{List of Protocols}
\SetKwInOut{Input}{Input}
\SetKwData{cn}{count}
\SetKwFunction{Count}{COUNT}

\newcommand{\x}{x}

\newcommand{\condd}[2]{\left| \begin{array}{c} #1 \\ #2 \end{array} \right.}

\newcommand{\xx}[1]{{\bf x_{#1}}}
\newcommand{\II}[2]{{\bf I_{#1}^{#2}}}

\newcommand{\green}[1]{\textnormal{G}\textsc{reen}_{#1}}
\newcommand{\yellow}{\textnormal{Y}\textsc{ellow}}
\newcommand{\purple}[1]{\textnormal{P}\textsc{urple}_{#1}}
\newcommand{\red}[1]{\textnormal{R}\textsc{ed}_{#1}}
\newcommand{\cyan}[1]{\textnormal{C}\textsc{yan}_{#1}}

\newcommand{\bfA}{\mathbf{A}}
\newcommand{\bfB}{\mathbf{B}}
\newcommand{\bfC}{\mathbf{C}}

\global\long\def\pull{\mathcal{PULL}}

\definecolor{dark_green}{rgb}{0,0.4,0}

\begin{document}

\title{Early Adapting to Trends: Self-Stabilizing Information Spread using Passive Communication\thanks{This work has received funding from the European Research Council (ERC) under the European Union's Horizon 2020 research and innovation program (grant  agreement No 648032).}}

\author{Amos Korman\thanks{CNRS, located at the French-Israeli Laboratory on Foundations of Computer Science, UMI FILOFOCS, CNRS, UP7, TAU, HUJI, WIS International Joint Research Unit, Tel-Aviv, Israel.} ~and Robin Vacus\thanks{CNRS, located at the Research Institute on the Foundations of Computer Science (IRIF), Paris, France.}}

\maketitle

\begin{abstract}
How to efficiently and reliably spread information in a system is one of the most fundamental problems in distributed computing. Recently, inspired by biological scenarios, several works focused on identifying the minimal communication resources necessary to spread information under faulty conditions. Here we study the  self-stabilizing \emph{bit-dissemination} problem, introduced by Boczkowski, Korman, and Natale in~[SODA 2017]. The problem considers a fully-connected network of $n$ \emph{agents}, with a binary world of \emph{opinions}, one of which is called \emph{correct}. At any given time, each agent holds an opinion bit as its public output. The population contains a \emph{source} agent which knows which opinion is correct. This agent adopts the correct opinion and remains with it throughout the execution. We consider the basic $\pull$ model of communication, in which each agent observes  relatively few randomly chosen agents in each round. The goal of the non-source agents is to quickly converge on the correct opinion, despite having an arbitrary initial configuration, i.e., in a self-stabilizing manner. Once the population converges on the correct opinion, it should remain with it forever. Motivated by biological scenarios in which animals observe and react to the behavior of others, we focus on the extremely constrained model of \emph{passive communication}, which assumes that when observing another agent the only information that can be extracted is the opinion bit of that agent. We prove that this problem can be solved in a poly-logarithmic in~$n$ number of rounds with high probability, while sampling a logarithmic number of agents at each round. Previous works solved this problem faster and using fewer samples, but they did that by decoupling the messages sent by agents from their output opinion, and hence do not fit the framework of passive communication. Moreover, these works use complex recursive algorithms with refined clocks that are unlikely to be used by biological entities. In contrast, our proposed algorithm has a natural appeal as it is based on letting agents estimate the current tendency direction of the dynamics, and then adapt to the emerging trend. 
\end{abstract}



\section{Introduction}

\subsection{Background and motivation}
Disseminating information from one or several sources to the whole
population is a fundamental building block in a myriad of distributed
systems \cite{censor2012global,demers1987epidemic,karp2000randomized,chierichetti2018rumor,giakkoupis2011randomness}, including
in multiple natural systems \cite{sumpter2008consensus,razin2013desert,feinerman2017individual}. 
This task becomes particularly challenging when the system is prone to faults~\cite{DBLP:journals/ploscb/BoczkowskiNFK18,DBLP:journals/jacm/GeorgiouGGK13,DBLP:journals/dc/GeorgiouGK11,dutta2013complexity}, or when the agents or their interactions are constrained \cite{DBLP:journals/dc/BoczkowskiKN19,angluin2006computation}. These issues find relevance in a variety of systems, including insect populations~\cite{razin2013desert}, chemical reaction networks \cite{chen2014speed}, and mobile~sensor~networks \cite{yick2008wireless}. In particular, in many biological systems, the internal computational abilities of individuals are impressively diverse, whereas the communication capacity is highly limited \cite{barclay1982interindividual,razin2013desert,feinerman2017individual}. An extreme situation, often referred to as {\em passive communication} \cite{wilkinson1992information}, is when information is gained by observing the behavior of other animals, which, in some cases, may not even intend to communicate \cite{cvikel2015bats,giraldeau2018social}. 
Such public information can reflect on the quality of possible behavioral options, hence allowing to improve fitness when used properly~\cite{danchin2004public}.

Consider, for example, the following scenario that serves as an inspiration for our model. A group of $n$ animals is scattered around an area searching for food. Assume that one side of the area, say, either the eastern side or the western side, is preferable (e.g., because it contains more food or because it is less likely to attract predators). However, only a few animals know which side is preferable. These knowledgeable animals will therefore spend most of their time in the preferable side of the area. Other animals would like to exploit the knowledge held by the knowledgeable animals, but they are unable to distinguish them from others. Instead, what they can do, is to scan the area in order to roughly estimate how many animals are on each side, and, if they wish, move between the two sides. Can the group of non-knowledgeable animals manage to locate themselves on the preferable side relatively fast, despite initially being spread in an arbitrary way while being completely uncorrelated?

The scenario above illustrates the notion of passive communication. The decision that an animal must make at any given time is to specify on which side of the area it should forage. This choice would be visible by others, and would in fact be the only information that an animal could reveal. Moreover, it cannot avoid revealing it. In particular, even the knowledgeable animals, who do not necessarily wish to communicate intentionally, cannot avoid revealing their ``correct'' choice. Assuming these animals do not actively try to harm others, they would simply reside on the preferable side, and promote this choice passively. The unknowledgeable animals, on the other hand, have a clearer incentive to cooperate, and they could, in principle, manipulate their choices to enhance the convergence process towards the correct choice. However, if their algorithm is also required to self-stabilize, then such a manipulation would be highly limited because, eventually, the choices should converge on one particular choice.

\subsection{The problem}
This paper studies the {\em self-stabilizing bit-dissemination} problem, introduced by Boczkowski, Korman, and Natale in \cite{DBLP:journals/dc/BoczkowskiKN19}, with the aim of solving it using passive communication. The problem considers a fully-connected network of $n$ agents, and a binary world of {\em opinions}, say $\{0,1\}$. One of these opinions is called {\em correct} and the other is called {\em wrong}. Execution proceeds in synchronous rounds (though agents do not have knowledge about the round number). At any given round~$t$, each agent $i$ holds an opinion bit $Y_t^{(i)}\in \{0,1\}$ (viewed as its output variable). The population contains one {\em source agent} which knows which opinion is correct. This agent adopts the correct opinion and remains with it throughout the execution.
Each agent knows whether or not it is the source, which can be formalized by assuming a designated {\em source-bit} in the state of an agent indicating this fact. We study the basic $\pull$ model of communication \cite{demers1987epidemic,karp2000randomized,DBLP:journals/ploscb/BoczkowskiNFK18}, in which in each round, each agent sees the information held by $\ell$ other agents, chosen uniformly at random (with replacement), where $\ell$ is small compared to $n$. 
 We consider the {\em passive communication} model which assumes that the only information that can be obtained by sampling an agent is its opinion bit. Hence, sampling $\ell$ agents is equivalent to receiving an integer between $0$ and $\ell$ corresponding to the number of agents with opinion 1 among the sampled agents.  
 
 In the {\em self-stabilizing} framework, the goal of the non-source agents is to quickly converge on the correct opinion, despite having an arbitrary initial configuration, that is set by an adversary (and despite not being able to distinguish the source from non-sources). We note that although the initial states of agents are arbitrary, we assume that their source-bit is not corrupted, and therefore, they reliably know whether or not they are the source. Moreover, by our assumption, it is guaranteed that there is only one source agent in the system\footnote{Our framework and proofs can be extended to allow for a constant number of sources, however, in this case it must be guaranteed that all sources agree on which opinion is the correct one. Indeed, as mentioned below, when there are conflicts between sources, the problem cannot be solved efficiently in the passive communication model, even if significantly more agents support one opinion.}. The adversary may initially set a different opinion to the source, but then the value of the correct bit would change, and the convergence should be guaranteed with respect to the new value.


The {\em running time} of the protocol corresponds to the first round $t_{con}$ that the configuration of opinions reached a consensus on the correct opinion, and remained unchanged forever after. We say that a protocol converges in time $T$ {\em with high probability (w.h.p)} if $t_{con}\leq T$ with probability at least $1-1/n^c$, for some constant $c>1$. 
Note that we do not require agents to irrevocably commit on their final opinion, but rather that they eventually converge on the correct opinion without necessarily being aware that convergence has happened. 

\paragraph{\bf On the difficulties resulting from using passive communication.} 
Previous works on the self-stabilizing bit-dissemination problem focused on identifying the minimal number of bits per interaction (message size) that need to be revealed in order to solve the problem in a short time. Boczkowski, Korman, and Natale~showed in \cite{DBLP:journals/dc/BoczkowskiKN19} that the  problem can efficiently be solved in $\tilde{O}(\log n)$ rounds w.h.p, by sampling $\ell=2$ agents at each round and using messages of size 3 bits. The protocol therein is based on agents having clocks that tick at each round. However, because of the self-stabilizing setting, the values of these clocks may initially be completely uncorrelated. The idea in that paper was to use 3 auxiliary bits in the messages to synchronize the clocks of agents in a self-stabilizing manner. These clocks were then used to facilitate the convergence of the opinion bit, which encapsulated another bit in the message. These 4 bits were then compressed to 3 bits using a recursive message-size-reduction mechanism. By reducing the message size of that clock-synchronization scheme to 1,~and using a version of the aforementioned recursive message-size-reduction mechanism, a recent work by Bastide, Giakkoupis, and Saribekyan allows to reduce the message size of the self-stabilizing bit-dissemination problem to 1 bit \cite{bastide2021self}. Moreover, stabilization is achieved in $O(\log n)$ rounds w.h.p., by sampling a single agent at each round. Importantly, however, both of these works decoupled the opinion of an agent from the message it uses, and hence, do not fit the framework of passive communication. Moreover, 
 these works use complex recursive algorithms with refined clocks that are unlikely to be used by biological entities. Instead, we are interested in identifying algorithms that have a more natural appeal.

To illustrate the difficulty of self-stabilizing information spread under passive communication, let us consider a more generalized problem than bit-dissemination called 
{\em majority bit-dissemination}. In this problem, the population contains $k\geq 1$ source agents which may not necessarily agree on which opinion is correct.
Specifically, in addition to its opinion, each source-agent stores a {\em preference} bit~$\in \{0,1\}$. Let~$k_i$ be the number of source agents whose preference is~$i$.
Assume that sufficiently more source agents share a preference~$i$ over~$1-i$ (e.g., at least twice as many), and call~$i$ the {\em correct bit}. Then, w.h.p., all agents (including the sources that might have the opposite preference)
should converge their opinions on the correct bit in poly-logarithmic time, and remain with that opinion for polynomial time\footnote{Observe that the case $k=1$ consists of a slightly weaker version than the bit-dissemination problem because the latter problem requires that after convergence is guaranteed w.h.p.,~ the configuration remains correct for an indefinite time with probability~1.}. 
The authors of \cite{DBLP:journals/dc/BoczkowskiKN19} showed that the self-stabilizing majority bit-dissemination problem can be solved in logarithmic time, using messages of size 3 bits, and the authors of \cite{bastide2021self} showed how to reduce the message size to 1. As mentioned, the messages in these protocols were different than the opinions, which were stored as internal variables, and hence the protocols in \cite{DBLP:journals/dc/BoczkowskiKN19,bastide2021self} are not based on passive communication. In fact, the following simple argument implies that this problem could not be solved in poly-logarithmic time under the model of passive communication, even if the sample size is~$n$ (i.e., all agents are being observed in each round)!

Assume by contradiction that there exists a self-stabilizing algorithm that solves the majority bit-dissemination problem using passive communication. Let us run this algorithm on a scenario with $k_1=n/2$ and $k_0=n/4$. Since $k_1\gg k_0$, then after a poly-logarithmic time, w.h.p., all agents would have opinion 1, and would remain with that opinion for polynomial time. 
Denote by $t_0$ the first time after convergence, and
let $s$ denote the internal state of one of the $n/4$ non-source agents at time $t_0$. Similarly, let $s'$ denote the internal state at time $t_0$ of one of the $n/4$ source agents with preference~$0$.
Now consider a second scenario, where we have $k_0=n/4$ and $k_1=0$. 
An adversary sets the internal states of agents (including their opinions) as follows. The internal states of the $k_0$ source agents (with preference~$0$) are all set to be $s'$. Moreover, their opinions (that these sources must publicly declare on) are all~$1$.
Next, the adversary sets the internal states of all non-source agents to be $s$, and their opinions to be 1.  
We now compare the execution of the algorithm in the first scenario (starting at time $t_0$) with the execution in the second scenario (starting at time 0, i.e., after the adversary manipulated the states). Note that both scenarios start with all opinions being 1. Hence, since we consider the passive communication model, all observations in the first round of the corresponding executions, would be unanimously 1. Moreover, as long as no agent changes its opinion in both scenarios, all observations would continue to be unanimously 1.
Furthermore, it is given that from time $t_0$, w.h.p., all agents in the first scenario remain with opinion 1 for a polynomially long period. Therefore, using a union bound argument, it is easy to see that also in the second scenario, w.h.p., all agents would remain with opinion 1 for polynomial time. This contradicts the fact that in the second scenario, w.h.p., the agents should converge on the opinion 0 in poly-logarithmic time. 

Note that the aforementioned impossibility result does not preclude the possibility of solving the self-stabilizing bit dissemination problem in the passive communication model, which does not involve a conflict between sources. Indeed, the authors of \cite{DBLP:journals/dc/BoczkowskiKN19} have suggested several candidate protocols which worked well in simulations\footnote{Note that simulations results may be deceiving in self-stabilizing contexts, since the worst initial conditions for a given protocol are not always evident.}, however, as mentioned therein, their  
analysis appears to be beyond the reach of currently known techniques regarding randomly-interacting agent systems in self-stabilizing contexts. 

\subsection{Our results}
 We propose a simple algorithm that efficiently solves the self-stabilizing bit-dissemination problem in 
 the passive communication model. The algorithm has a natural appeal as it is based on letting agents estimate the current tendency direction of the dynamics, and then adapt to the emerging trend. More precisely (but still, informally), each non-source agent counts the number of agents with opinion 1 it observes in the current round and compares it to the number observed in the previous round. If more 1's are observed now, then the agent adopts the opinion 1, and similarly, if more 0's are observed now, then it adopts the opinion 0 (if the same number of 1's is observed in both rounds then the agent does not change its opinion). 
 Intuitively, on the global level, this behavior creates a persistent movement of the average opinion of the non-source agents towards either $0$ or $1$, which ``bounces'' back when hitting the wrong opinion.
 
 More formally, we first describe the following algorithm. In addition to the opinion bit, the algorithm at an agent uses two internal variables at round~$t$, called $\cn_t$ and $\cn_{t-1}$. The former variable aims to store the number of 1's observe in the previous round, while the latter aims to store the number of 1's observe in the current round. Hence, both encode a number in $\{ 0,\ldots, \ell\}$, and can be stored using $O(\log\ell)$ bits of memory. The following procedure is executed at each round separately by each agent.
As an input, a given agent observes the opinions of $\ell$ randomly chosen agents. Let $J_t$ denote the set of agents it sampled at round $t$, and let $S_t(J_t)=\{Y_t^{(j)}\}_{j\in J_t}$ denote the set containing their opinions at round $t$. Finally, for any set $A$ of opinions, let $\Count(A)$ denote the number of 1-opinions in $A$.

\setlength{\algotitleheightrule}{0pt}%
 \begin{algorithm}[htbp]
\Input{$S_t(J_t)$}
$\cn_t \leftarrow \Count(S_t(J_t))$ \;
\lIf{$\cn_t > \cn_{t-1}$} {
    $Y_{t+1} \leftarrow 1 $
} \lElseIf {$\cn_t < \cn_{t-1}$} {
    $Y_{t+1} \leftarrow 0 $
} \lElse {
    $Y_{t+1} \leftarrow Y_t $
}
\end{algorithm}
As it turns out, one feature of the aforementioned protocol will make the analysis difficult -- that is, that $Y_{t}$ and~$Y_{t+1}$ are dependent, even when conditioning on~$(\x_{t-1},\x_{t})$. This is because $\cn_{t-1}$ is used to compute both~$Y_{t}$ and~$Y_{t+1}$. 
For example, if the set $J_{t-1}$ sampled at round~$t-1$ happens to contain more 1's, then~$\cn_{t-1}$ is larger. In this case, $Y_{t}$ has a higher chance of being~$1$, and~$Y_{t+1}$ has a higher chance of being~$0$. For this reason we introduce a modified version of the protocol that solves this dependence issue. The idea is to partition the set of opinions sampled at round $t$ into 2 sets of equal size. One subset will be used to compare with a subset of round $t-1$, and the other subset will be used to compare with a subset of round $t+1$. Note that this implies that the set of agents sampled in a round is $2\ell$ rather than $\ell$, however, since we are interested in the case $\ell=O(\log n)$, this does not cause a problem. This modified protocol, called {\em Follow the Emerging Trend (FET)}, is the one we shall actually analyze.  
\setlength{\algotitleheightrule}{1pt}%
\begin{algorithm}[htbp]
    \caption{Follow the Emerging Trend (FET) at round~$t$}\label{protocol:FET}
    \Input{$S_t(J_t)$}
    Partition~$S_t(J_t)$ into two sets~$S'_{t},S''_{t}$ of equal size uniformly at random \;
    $\cnt_t \leftarrow \Count(S'_t)$ ; \quad $\cntt_t \leftarrow \Count(S''_t)$ \;
    \lIf{$\cnt_t > \cntt_{t-1}$} {
        $Y_{t+1} \leftarrow 1 $
    } \lElseIf {$\cnt_t < \cntt_{t-1}$} {
        $Y_{t+1} \leftarrow 0 $
    } \lElse {
        $Y_{t+1} \leftarrow Y_t $
    }
\end{algorithm}
Note that, although we used time indices for clarity, the protocol does not require the agents to know~$t$.
The following consists of the main result in the paper.

\begin{theorem} \label{thm:main}
    Algorithm FET solves the self-stabilizing bit-dissemination problem in the passive communication model. It converges in  $O(\log^{5/2} n)$ rounds on the correct opinion, with high  probability, while relying on $\ell=O(\log n)$ samples in each round, and using $O(\log \ell)$ bits of memory per agent.     
\end{theorem}

\subsection{Related works} 
In recent years, the study of {\em population protocols} has attracted significant attention  in the distributed computing community \cite{DBLP:journals/sigact/AlistarhG18,DBLP:conf/icalp/AlistarhG15,DBLP:journals/corr/AlistarhAEGR16,DBLP:journals/dc/AngluinAE08,DBLP:journals/eatcs/AspnesR07}.
These models often consider agents that interact under random meeting patterns while being restricted in both their memory and communication capacities. By now, we understand the computational power of such systems rather well, but apart from a few exceptions, this understanding is limited to non-faulty scenarios.

The framework of {\em opinion dynamics} corresponds to settings of multiple agents, where in each round, each agent samples one or more agents at random, extracts their opinions, and employs a certain rule for updating its opinion. The study of opinion dynamics crosses disciplines, and is highly active in physics and computer science, see review e.g., in \cite{becchetti2020consensus}. Many of the models of opinion dynamics can be considered as following passive communication, since the information an agent reveals coincides with its opinion. Generally speaking, however, the typical scientific approach in opinion dynamics is to start with some simple update rule, and analyze the resulting dynamics, rather than tailoring an updating rule to solve a given distributed problem. For example, perhaps the most famous dynamics in the context of interacting particles systems concerns the {\em voter} model~\cite{liggett1985interacting}. In theoretical computer science, extensive research has been devoted to analyzing the time to reach consensus, following different updating rules including the {\em 3-majority} \cite{doerr2011stabilizing}, {\em Undecided-State Dynamics} \cite{DBLP:journals/dc/AngluinAE08}, and others. In these works, consensus should be reached either on an arbitrary value, or on the majority (or plurality) opinion, as evident in the initial configuration. 

In many natural settings, however, the group must converge on a particular consensus value that is a function of the environment. Moreover, agents have different levels of knowledge regarding the desired value, and the system must utilize the information held by the more knowledgeable individuals \cite{sumpter2008consensus,DBLP:journals/fams/AyalonSFKGF21,DBLP:journals/ploscb/KormanGF14,rajendran2022ants}. As explained in more detail below, when communication is restricted, and the system is prone to faults, this task can become  challenging. 

Propagating information from one or more sources to the rest of the population has been the focus of a myriad of works in distributed computing.
This dissemination problem has been studied under various models taking different names, including {\em rumor spreading}, {\em information spreading}, {\em gossip}, {\em broadcast}, and others, see e.g., \cite{giakkoupis2014tight,censor2012global,demers1987epidemic,karp2000randomized,chierichetti2018rumor,giakkoupis2011randomness}. 
A classical algorithm in the $\pull$ model spreads the opinion of the source to all others in $2\log n$ rounds, by letting each uninformed agent copy
the opinion of an informed agent whenever seeing one for the first time \cite{karp2000randomized}. Unfortunately, this elegant algorithm does not suit all realistic scenarios, since its soundness crucially relies on the absence of misleading information. To address such issues, rumor spreading has been studied under different models of faults. One line of research investigates the case that messages may be corrupted with some fixed probability \cite{DBLP:journals/dc/FeinermanHK17,DBLP:conf/innovations/BoczkowskiFKN18}. 
Another model of faults is {\em self-stabilization} \cite{dijkstra}, 
where the system must eventually converge on the opinion of the source regardless of the initial configuration of states \cite{dijkstra}. For example, the algorithm in \cite{karp2000randomized} fails in this setting, since non-source agents may be initialized to ``think'' that they have already been informed by the correct opinion, while they actually hold the wrong opinion. 

\paragraph{\bf More details on previous works regarding the self-stabilizing bit-dissemination problem.}
The self-stabilizing bit-dissemination problem was introduced in \cite{DBLP:journals/dc/BoczkowskiKN19}, with the aim of minimizing the message size. As mentioned therein, if all agents share the same notion of global time, then convergence can be achieved in $\bigO(\log n)$ time w.h.p.~even under passive communication. The idea is that agents divide the time horizon into phases of length $T=4\log n$, and that each phase is further subdivided into $2$ subphases of length $2\log n$ each. In the first subphase of each phase, if a non-source agent observes an opinion $0$, then it copies it as its new opinion, but if it sees 1 it ignores it. In the second subphase, it does the opposite, namely, it adopts the output bit $1$ if and only if it sees an opinion $1$. Now, consider the first phase. If the source supports opinion 0 then at the end of the first subphase, every output bit would be $0$ w.h.p., and the configuration would remain that way forever. Otherwise, if the source supports 1, then at the end of the second subphase all output bits would be $1$ w.h.p., and remain $1$ forever. 


The aforementioned protocol indicates that the self-stabilizing bit-dissemination problem could be solved efficiently by running a self-stabilizing {\em clock-syncronization} protocol in parallel to the previous example protocol. This parallel execution amounts to adding one bit to the message size of the clock synchronization protocol. The main technical contribution of \cite{DBLP:journals/dc/BoczkowskiKN19}, as well as the focus of subsequent work in \cite{bastide2021self}, was solving the self-stabilizing clock-synchronization using as few as possible bits per message. In fact, the authors in \cite{bastide2021self} managed to do so using 1-bit messages. This construction thus implies a solution to the self-stabilizing bit-dissemination problem using 2 bits per message. A recursive procedure, similar to the one established in \cite{DBLP:journals/dc/BoczkowskiKN19}, then allowed to further compress the 2 bits into 1-bit messages. Importantly, however,  the 1-bit message revealed by an agent is different from its opinion bit, which is kept in the protocols of \cite{DBLP:journals/dc/BoczkowskiKN19,bastide2021self} as an internal variable.
At first glance, to adhere to the passive communication model, one may suggest that agents simply choose their opinion to be this 1-bit message used in \cite{bastide2021self}, just for the purpose of communication, until a consensus is reached, and then switch the opinion to be the correct bit, once it is identified. There are, however, two difficulties to consider regarding this approach. First, in our setting, the source agent  does not change its opinion (which, in the case of~\cite{bastide2021self}, may prevent the protocol from reaching a consensus at all). Second, even assuming that the protocol functions properly despite the source having a stable opinion, it is not clear how to transition from the ``communication'' phase (where agents use their opinion to operate the protocol, e.g., for synchronizing clocks) to the ``consensus'' phase (where all opinions must be equal to the correct bit at every round). For instance, the first agents to make the transition may disrupt other agents still in the first phase.

%


\section{Proof of Theorem \ref{thm:main}: General Overview}

The goal of this section is to prove Theorem~\ref{thm:main}. The $O(\log \ell)$ bits upper  bound on the memory complexity clearly follows from the fact that the only variables kept by the FET algorithm (Protocol \ref{protocol:FET}) are~$\cnt$ and~$\cntt$, which are used to count the number of 1's in a sample (of size $\ell$).

Since the protocol is symmetric with respect to the opinion of the source, we may assume without loss of generality that the source has opinion~$1$. Our goal would therefore be to show that the FET algorithm converges to 1 fast, w.h.p., regardless of the initial configuration of non-source agents.
Note that in order to achieve running time of $O(T)$ w.h.p guarantee, is it sufficient to show that the 
algorithm stabilizes in $T$ rounds with probability at least $1-1/n^\epsilon$, for some $\epsilon>0$. Indeed, because of the self-stabilizing property of the algorithm, the probability that the algorithm does not stabilize within $2T/\epsilon$ rounds is at most $(1/n^\epsilon)^{2/\epsilon}=1/n^{2}$. 


For the sake of analysis, let~$\x_t$ denote the fraction of agents with opinion~$1$ at round $t$ among the whole population of agents (including the source). We shall extensively use the two dimensional grid $\calG :=\{0,\frac{1}{n},\ldots,\frac{n-1}{n},1\}^2$. When analyzing what happens at round $t+2$, the $x$-axis of $\calG$ would represent $\x_t$, and the $y$-axis would represent $\x_{t+1}$.


\begin{observation} \label{lem:evolution}
    For any round~$t$, conditioning on $\x_t = \xx{t}$, and $\x_{t+1} = \xx{t+1}$,
    the probability that a non-source agent~$i$ has opinion~$1$ on round~$t+2$ is
    \begin{equation} \label{eq:indiv_next_opinion}
        \bbP \pa{ Y_{t+2}^{(i)} = 1 \condd{\x_t = \xx{t}}{\x_{t+1} = \xx{t+1}} }  =  \prob{\ell}{\xx{t+1}}{>}{\ell}{\xx{t}} + \bbOne_{\{ Y_{t+1}^{(i)} = 1 \}} \cdot \prob{\ell}{\xx{t+1}}{=}{\ell}{\xx{t}}.
    \end{equation}
    Moreover, there are independent binary random variables~$X_1,\ldots,X_n$ such that $\x_{t+2}$ is distributed as~$\frac{1}{n} \sum_{i=1}^n X_i$.
    Eventually,
    \begin{equation} \label{eq:next_expectation}
         \bbE \pa{ \x_{t+2} \condd{\x_t = \xx{t}}{\x_{t+1} = \xx{t+1}} } =
         \prob{\ell}{\xx{t+1}}{>}{\ell}{\xx{t}} + \xx{t+1} \cdot \prob{\ell}{\xx{t+1}}{=}{\ell}{\xx{t}} + \frac{1}{n}(1-\prob{\ell}{\xx{t+1}}{\geq}{\ell}{\xx{t}}).
    \end{equation}
\end{observation}

The proof of Observation~\ref{lem:evolution} is deferred to Appendix~\ref{app:evolution}.
A consequence of Observation~\ref{lem:evolution}, is that the execution of the algorithm induces a Markov chain on~$\calG$.
This Markov chain has a unique absorbing state, $(1,1)$, since we assumed the source to have opinion~$1$.
To prove Theorem~\ref{thm:main} we therefore only need to bound the time needed to reach~$(1,1)$.

\subsection{Partitioning the grid into domains}
Let us fix~$0<\delta<1/2$ and $\lambda_n = \frac{1}{\log^{1/2+\delta}n}$. We partition $\calG$ into domains as follows (see illustration on Figure~\ref{fig:partition}).
\begin{align*}
    \green{1} &= \bigg\{ (\x_t,\x_{t+1}) \text{ s.t. } \x_{t+1} \geq \x_t + \delta \bigg\}, \\
    \purple{1} &= \bigg\{ (\x_t,\x_{t+1}) \text{ s.t. } \frac{1}{\log n} \leq \x_t < \frac{1}{2}-3\delta \text{ and } (1-\lambda_n) \cdot \x_t \leq \x_{t+1} < \x_t + \delta \bigg\}, \\
    \red{1} &= \bigg\{ (\x_t,\x_{t+1}) \text{ s.t. } \frac{1}{\log n} \leq \x_{t+1} \text{ and } \x_t < \frac{1}{2}-3\delta \text{ and } \x_t - \delta \leq \x_{t+1} < (1-\lambda_n) \cdot \x_t \bigg\}, \\
    \cyan{1} &= \bigg\{ (\x_t,\x_{t+1}) \text{ s.t. } 0 \leq \min(\x_t,\x_{t+1}) < \frac{1}{\log n} \text{ and } \x_t - \delta < \x_{t+1} < \x_t + \delta \bigg\}, \\
    \yellow &= \bigg\{ (\x_t,\x_{t+1}) \text{ s.t. } \frac{1}{2} - 3\delta \leq \x_t < \frac{1}{2} \leq 3\delta \text{ and } \frac{1}{2} - 4\delta \leq \x_{t+1} \leq \frac{1}{2} + 4\delta \text{ and } |\x_{t+1}-\x_t| < \delta  \bigg\}.
\end{align*}
Similarly, for the former 4 domains, we define~$\green{0},\purple{0},\red{0}$ and~$\cyan{0}$ to be their symmetric equivalents (w.r.t the point ($\frac{1}{2},\frac{1}{2}$)), and finally define:
    $\green{} = \green{0} \cup \green{1}$, 
    $\purple{} = \purple{0} \cup \purple{1}$,
    $\red{} = \red{0} \cup \red{1}$, and
    $\cyan{} = \cyan{0} \cup \cyan{1}$.
We shall analyze each area separately, conditioning on the Markov chain to be at any point in that area, and focusing on the number of rounds required to escape the area, and the probability that this escape is made to a particular other area. 
Figure~\ref{fig:main_proof} represents a sketch of the proof of Theorem~\ref{thm:main}, which may help to navigate the intermediate results. 
\begin{figure}
    \centering
    \begin{subfigure}[b]{0.55\textwidth}
        \centering
        \includegraphics[width=\textwidth]{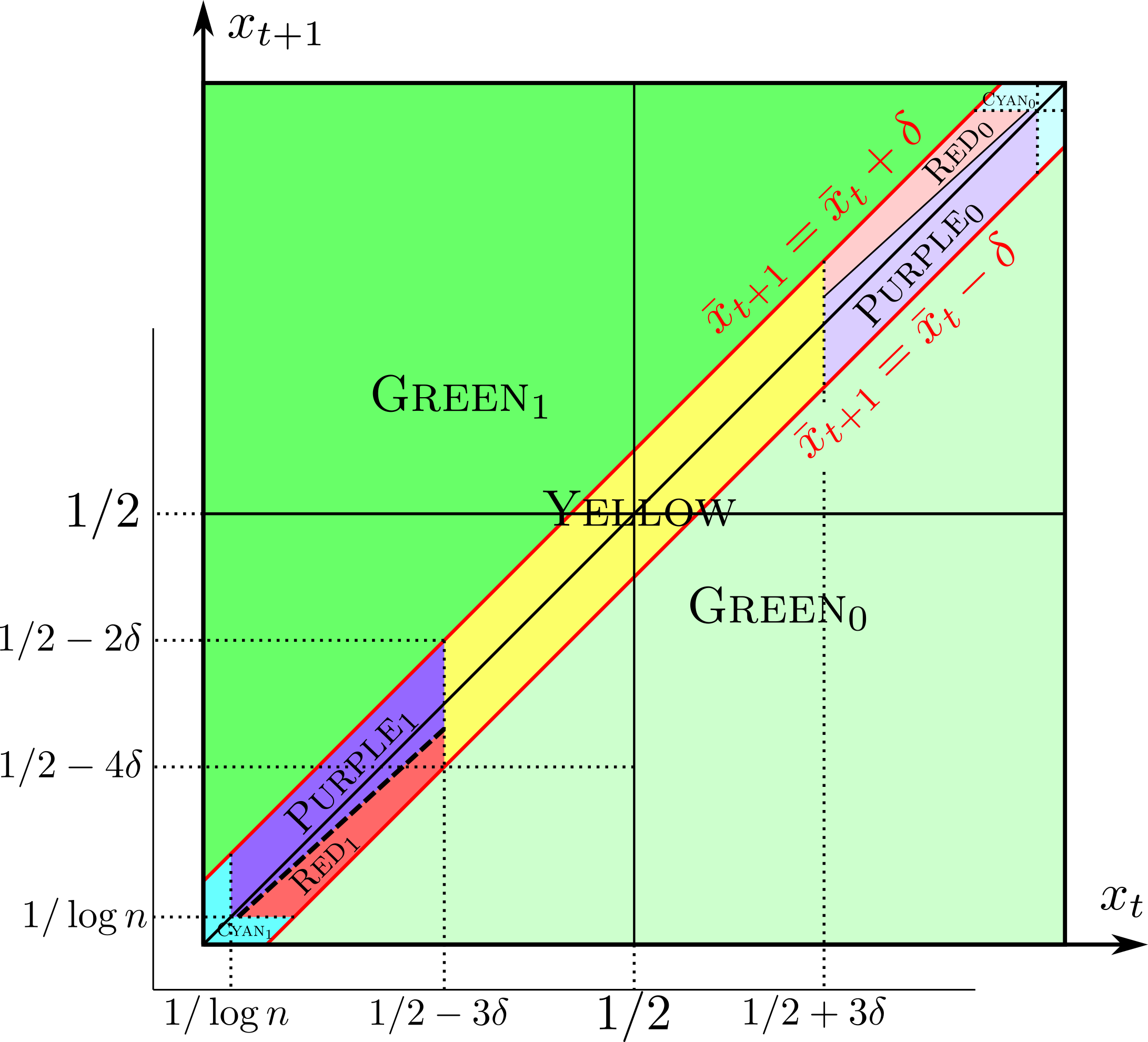}
        \caption{}
        \label{fig:partition}
    \end{subfigure}
    \hfill
    \begin{subfigure}[b]{0.4\textwidth}
        \centering
        \includegraphics[width=\textwidth]{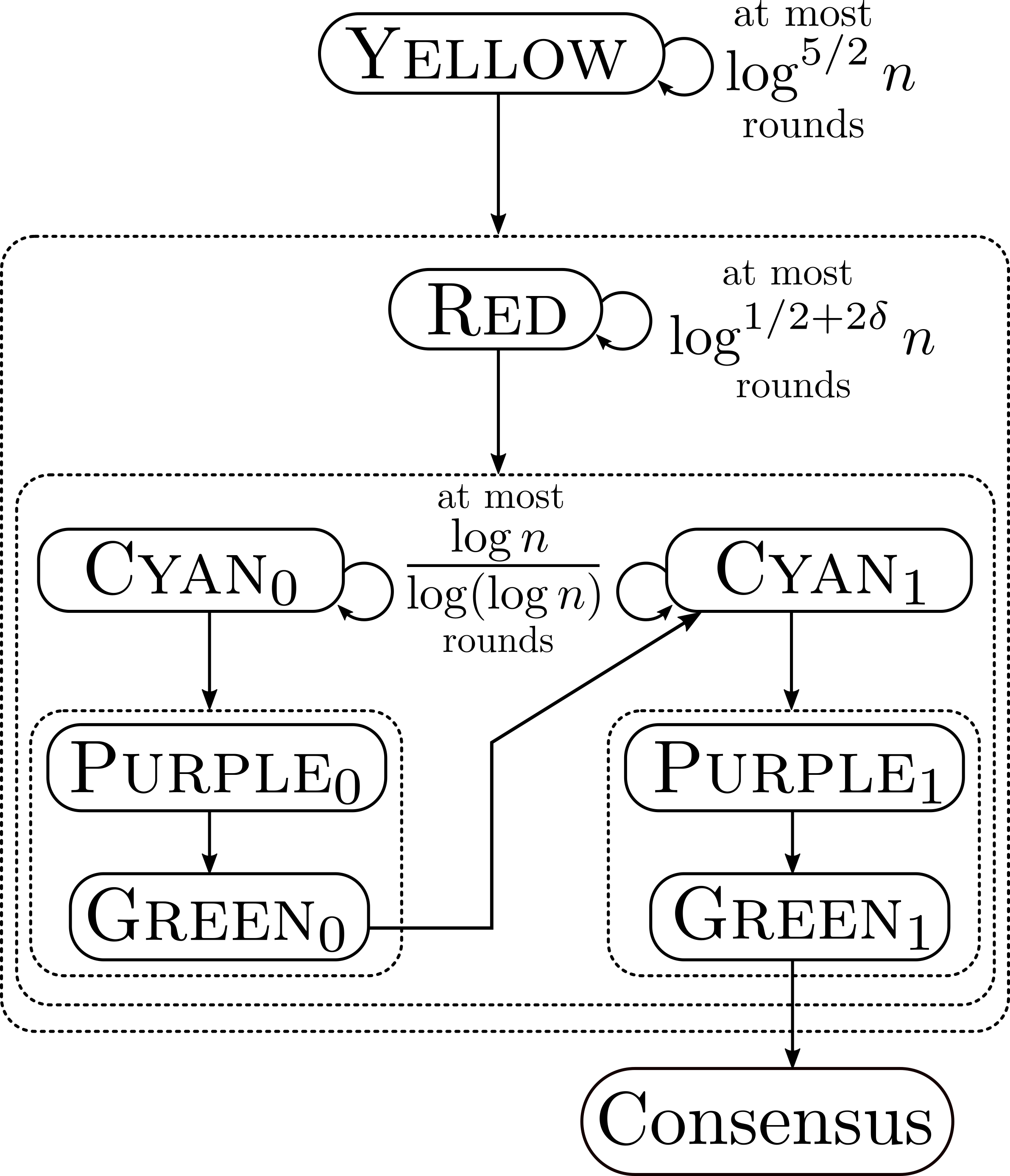}
        \caption{}
        \label{fig:main_proof}
    \end{subfigure}
    \caption{(a) Partitioning the state space into domains. The x-axis (resp., y-axis) represents the proportion of agents with opinion~$1$ in round~$t$ (resp.,~$t+1$). The thick dashed line at the frontier between~$\purple{1}$ and~$\red{1}$ is defined by~$\x_{t+1} = (1-\lambda_n) \x_t$. (b) Sketch of the proof of Theorem~\ref{thm:main}. All transitions are  w.p. at least~$1-1/n^{\Omega(1)}$. The numbers next to the self-loops indicate that the process stays in the corresponding domain for this number of rounds w.p. at least~$1-1/n^{\Omega(1)}$. The source is assumed to have opinion~$1$.}
\end{figure}

As it happens, the dynamics starting from a point $(\x_t,\x_{t+1})$ highly depends on the difference between  $\x_t$ and~$\x_{t+1}$. Roughly speaking, the larger $|\x_{t+1}-\x_t|$ is the faster is the convergence. For this reason, we refer to $|\x_{t+1}-\x_t|$ as the {\em speed} of the point $(\x_t,\x_{t+1})$. (This could also be viewed as the ``derivative'' of the process at time $t$.)

\subsection{Analyzing the Markov chain at different domains}
Due to lack of space, most proofs are deferred to the appendix. We have decided to include in the extended abstract the core of the proof of the main lemma that concerns the Yellow domain (Section~\ref{sec:lem:B1}), which was one of the more challenging results obtained in this paper. In addition, we also provide the core of the proof of the main lemma that concerns the Cyan domain (Section~\ref{sec:lem:cyan}), because this part of the state-space is essential for understanding the dynamics of the protocol. 

Let us now give an overview of the intermediate results. 
First we consider~$\green{}$, in which the speed of points is large. In Lemma \ref{lem:green} we show that from points in that domain, non-source agents reach a consensus in just one round w.h.p. In particular, if the Markov chain is at some point in $\green{1}$, then the consensus will be on 1, and we are done. If, on the other hand, the Markov chain is in $\green{0}$, then the consensus of non-source agents would be on 0. As we show later, in that case the Markov chain would reach~$\cyan{1}$ in one round w.h.p.

\begin{lemma} [Green area] \label{lem:green}
    Assume that~$c$ is sufficiently large.
    If~$(\x_t,\x_{t+1}) \in \green{1}$, then w.h.p., for every non-source agent~$i$, $Y_i^{(t+2)} = 1$. Similarly, if~$(\x_t,\x_{t+1}) \in \green{0}$, then w.h.p., for every non-source agent~$i$, $Y_i^{(t+2)} = 0$.
\end{lemma}
The proof of Lemma~\ref{lem:green} follows from a simple application of H{\oe}ffding's inequality, and is deferred to Appendix~\ref{app:green}.
Next, we consider the area~$\purple{}$, and show that the population goes from~$\purple{}$ to~$\green{}$ in just one round, w.h.p.
In~$\purple{}$, the speed is relatively low, and~$\x_t$ and $\x_{t+1}$ are quite far from~$1/2$. On the next round, we expect~$\x_{t+2}$ to be close to~$1/2$, thus gaining enough speed in the process to join~$\green{}$. The proof of the following lemma is rather straightforward, and is deferred to Appendix~\ref{app:purple}.

\begin{lemma} [Purple area] \label{lem:purple}
    Assume that~$c$ is sufficiently large.
    If~$(\x_t,\x_{t+1}) \in \purple{1}$, then~$(\x_{t+1},\x_{t+2}) \in \green{1}$ w.h.p.
    Similarly, if~$(\x_t,\x_{t+1}) \in \purple{0}$, then~$(\x_{t+1},\x_{t+2}) \in \green{0}$ w.h.p.
\end{lemma}

Next, we bound the time that can be spent in~$\red{}$, by using the fact that as long as the process is in~$\red{1}$ (resp., $\red{0}$), $\x_t$ (resp., $(1-\x_t)$) decreases (deterministically) by at least a multiplicative factor of~$(1-\lambda_n)$ at each round.
After a poly-logarithmic number of rounds, the Markov chain must leave~$\red{}$ and in this case, we can show that it cannot reach~$\yellow$ right away. The proof of the following lemma is again relatively simple, and is deferred to Appendix~\ref{app:red}.

\begin{lemma} [Red area] \label{lem:red}
    Consider the case that~$(\x_{t_0},\x_{t_0+1}) \in \red{}$ for some round~$t_0$, and let~$t_1 = \min \{ t \geq t_0, (\x_t,\x_{t+1}) \notin \red{} \}$. Then~$t_1 < t_0 + \log^{1/2+2\delta} n$, and~$(\x_{t_1},\x_{t_1+1}) \notin \yellow \cup \red{}$.
\end{lemma}

Next, we bound the time that can be spent in~$\cyan{1}$. (A similar result holds for $\cyan{0}$.)
Roughly speaking, this area corresponds to the situation in which, over the last two consecutive rounds, the population is in an almost-consensus over the wrong opinion.
In this case, many agents (a constant fraction) see only 0 in their corresponding samples in the latter round. As a consequence, everyone of them who will see at least one opinion 1 in the next round, will adopt opinion 1. We can expect this number to be of order~$\ell=O(\log n)$. This means that, as long as the Markov chain is in~$\cyan{1}$, the value of $\x_t$ would grow by a logarithmic factor in each round. This implies that within $\log(n)/\log(\log n)$ rounds, the Markov chain will leave the $\cyan{1}$ area and go to $\green{1} \cup \purple{1}$. 
Informally, this phenomenon can be viewed as a form of ``bouncing'' --- the population of non-sources reaches an almost consensus on the wrong opinion, and ``bounces back'', by gradually increasing the fraction of agents with the correct opinion, up to an extent that is sufficient to enter $\green{1} \cup \purple{1}$.
The proof of the following lemma 
is given in Section~\ref{sec:lem:cyan}.

\begin{lemma} [Cyan area] \label{lem:cyan}
    Consider the case that~$(\x_{t_0},\x_{t_0+1}) \in \cyan{1}$ for some round~$t_0$, and let~$t_1 = \min \{ t \geq t_0, (\x_t,\x_{t+1}) \notin \cyan{1} \}$. Then with probability at least $1-1/n^{\Omega(1)}$ we have (1)~$t_1 < t_0 + \log(n)/\log(\log n)$, and (2) ~$(\x_{t_1},\x_{t_1+1}) \in \green{1} \cup \purple{1}$.
\end{lemma}

Eventually, we consider the central area, namely, $\yellow$, where the speed is very low, and bound the time that can be spent there. The proof of the following lemma is more complex than the previous ones, and it appears in Section~\ref{sec:yellow}.
\begin{lemma} [Yellow area] \label{lem:yellow}
    Consider the case that~$(\x_{t_0},\x_{t_0+1}) \in \yellow$. Then, w.h.p.,
    \begin{equation*}
        \min \{ t > t_0 \text{ s.t. } (\x_t,\x_{t+1}) \notin \yellow \} < t_0 + O(\log^{5/2} n). 
    \end{equation*}
\end{lemma}
\subsection{Assembling the lemmas}
Given the aforementioned lemmas, we have everything we need to prove our main result.

\begin{proof} [Proof of Theorem~\ref{thm:main}]
    Recall that without loss of generality, we assumed the source to have opinion~$1$.
    The reader is strongly encouraged to refer to Figure~\ref{fig:main_proof} to follow the ensuing arguments more easily.
    \begin{itemize}
        \item Let~$t_1 = \min \{ t \geq 0, (\x_t,\x_{t+1}) \notin \yellow \}$. If~$(\x_0,\x_1) \in \yellow$, we apply Lemma~\ref{lem:yellow} to get that
        \begin{equation} \label{property1}
            t_1 < O(\log^{5/2} n) ~\text{w.h.p. and } (\x_{t_1},\x_{t_1+1}) \in \red{} \cup \cyan{} \cup \purple{} \cup \green{}.
        \end{equation} 
        Else, $(\x_0,\x_1) \notin \yellow$ so~$t_1 = 0$, and~Eq.~\eqref{property1} also holds.
        \item Let~$t_2 = \min \{ t \geq t_1, (\x_t,\x_{t+1}) \notin \red{} \}$.
        If~$(\x_{t_1},\x_{t_1+1}) \in \red{}$, we apply Lemma~\ref{lem:red} to get that
        \begin{equation} \label{property2}
            t_2 < t_1 + \log^{1/2 + 2\delta} n ~\text{w.h.p. and } (\x_{t_2},\x_{t_2+1}) \in \cyan{} \cup \purple{} \cup \green{}.
        \end{equation}
        Else, $(\x_{t_1},\x_{t_1+1}) \notin \red{}$ so~$t_1 = t_2$, and by Eq.~\eqref{property1}, it implies that~Eq.~\eqref{property2} also holds.
        \item Let~$t_3 = \min \{ t \geq t_2, (\x_t,\x_{t+1}) \notin \cyan{} \}$. If~$(\x_{t_2},\x_{t_2+1}) \in \cyan{}$, we apply Lemma~\ref{lem:cyan} to get that
        \begin{equation} \label{property3}
            t_3 < t_2 + \log(n)/\log(\log n) ~\text{ and } (\x_{t_3},\x_{t_3+1}) \in \purple{} \cup \green{} \text{ with probability at least~$1-1/n^{\Omega(1)}$}.
        \end{equation}
        Else, $(\x_{t_2},\x_{t_2+1}) \notin \cyan{}$ so~$t_2 = t_3$, and by~Eq.~\eqref{property2}, it implies that~Eq.~\eqref{property3} also holds.
        \item Let~$t_4 = \min \{ t \geq t_3, (\x_t,\x_{t+1}) \in \green{} \}$. By Lemma~\ref{lem:purple}, and by~Eq.~\eqref{property3}, we have that~$t_4 = t_3$ or~$t_4 = t_3+1$ w.h.p.
    \end{itemize}
      If~$(\x_{t_4},\x_{t_4+1}) \in \green{1}$, then by Lemma~\ref{lem:green} the consensus is reached on round~$t_4+1$.
             Otherwise, if~$(\x_{t_4},\x_{t_4+1}) \in \green{0}$, by Lemma~\ref{lem:green}, we obtain that~$\x_{t_4+2} = 1/n$ w.h.p. (meaning that all agents have opinion~$0$ except the source). Therefore, in this case, either~$(\x_{t_4+1},\x_{t_4+2}) \in \green{0}$ or~$(\x_{t_4+1},\x_{t_4+2}) \in \cyan{1}$ (because for a point~$(\x_t,\x_{t+1})$ to be in any other area, it must be the case that~$\x_{t+1} \geq 1/ \log(n)$, by definition). In the former case, we apply Lemma~\ref{lem:green} again to get that~$\x_{t_4+3} = 1/n$ w.h.p., which implies that~$(x_{t_4+2},\x_{t_4+3}) = (1/n,1/n) \in \cyan{1}$.
             As we did before, we apply Lemma~\ref{lem:cyan}, \ref{lem:purple} and~\ref{lem:green} to show that, with probability at least~$1-1/n^{\Omega(1)}$, 
             the system goes successively to~$\purple{1} \cup \green{1}$, then to~$\green{1}$, and eventually reaches the absorbing state~$(1,1)$ in less than~$\log(n)/\log(\log n)+2$ rounds. 
             
             Altogether, the convergence time is dominated by $t_1$, and is hence $O(\log n)^{5/2}$ with probability at least~$1-1/n^{\epsilon}$, for some $\epsilon>0$. As mentioned, this implies that for any given $c>1$, the algorithm reaches consensus in $O(\log n)^{5/2}$ time with probability at least~$1-1/n^{c}$. This concludes the proof of~Theorem~\ref{thm:main}. 
\end{proof}

\section{Escaping the Yellow Area} \label{sec:yellow}

The goal of this section is to prove Lemma~\ref{lem:yellow}.
It might be easier for the reader to think of the Yellow area as a square. 
Formally, let us define~$\yellow'$ as the following square bounding box around~$\yellow$:
\begin{equation*}
    \yellow' = \bigg\{ (\x_t,\x_{t+1}) \text{ s.t. } 1/2-4\delta \leq \x_t,\x_{t+1} \leq 1/2+4\delta \bigg\}. 
\end{equation*}
Obviously, $\yellow \subset \yellow'$, so in order to prove Lemma~\ref{lem:yellow} it suffices to prove 
Lemma~\ref{lem:yellow2} below.
\begin{lemma} \label{lem:yellow2}
    Consider  that~$(\x_{t_0},\x_{t_0+1}) \in \yellow'$. Then, w.h.p.,
       $ \min \{ t > t_0 \text{ s.t. } (\x_t,\x_{t+1}) \notin \yellow' \} < t_0 + O(\log^{5/2} n).$ 
\end{lemma}

\subsection{General structure of the proof}
In order to prove Lemma \ref{lem:yellow2}, we first partition~$\yellow'$, as follows (for an illustration, see  Figure~\ref{fig:partition_yellow}):
\begin{align*}
    \bfA_1 &= \{ (\x_t,\x_{t+1}) \mid \text{ (i) } \x_{t+1} \geq 1/2 \text{ and (ii) } \x_{t+1}-\x_t  \geq \x_{t}-1/2 \} \medcap \yellow', \\
    \bfB_1 &= \{ (\x_t,\x_{t+1}) \mid \text{ (i) } \x_{t+1} \geq \x_t \text{ and (ii) } \x_{t+1}-\x_t <\x_{t}-1/2 \} \medcap \yellow', \\
    \bfC_1 &= \{ (\x_t,\x_{t+1}) \mid \text{ (i) } \x_{t+1}<1/2 \text{ and (ii) } \x_{t+1} \geq \x_{t} \} \medcap \yellow'.
\end{align*}
Similarly, we define~$\bfA_0,\bfB_0,\bfC_0$ their symmetric equivalents (w.r.t the point ($\frac{1}{2},\frac{1}{2}$)), and~$\bfA = \bfA_0 \cup \bfA_1$, $\bfB = \bfB_0 \cup \bfB_1$, and~$\bfC = \bfC_0 \cup \bfC_1$.
\begin{figure}[htbp]
    \centering
    \includegraphics[width=0.35\linewidth]{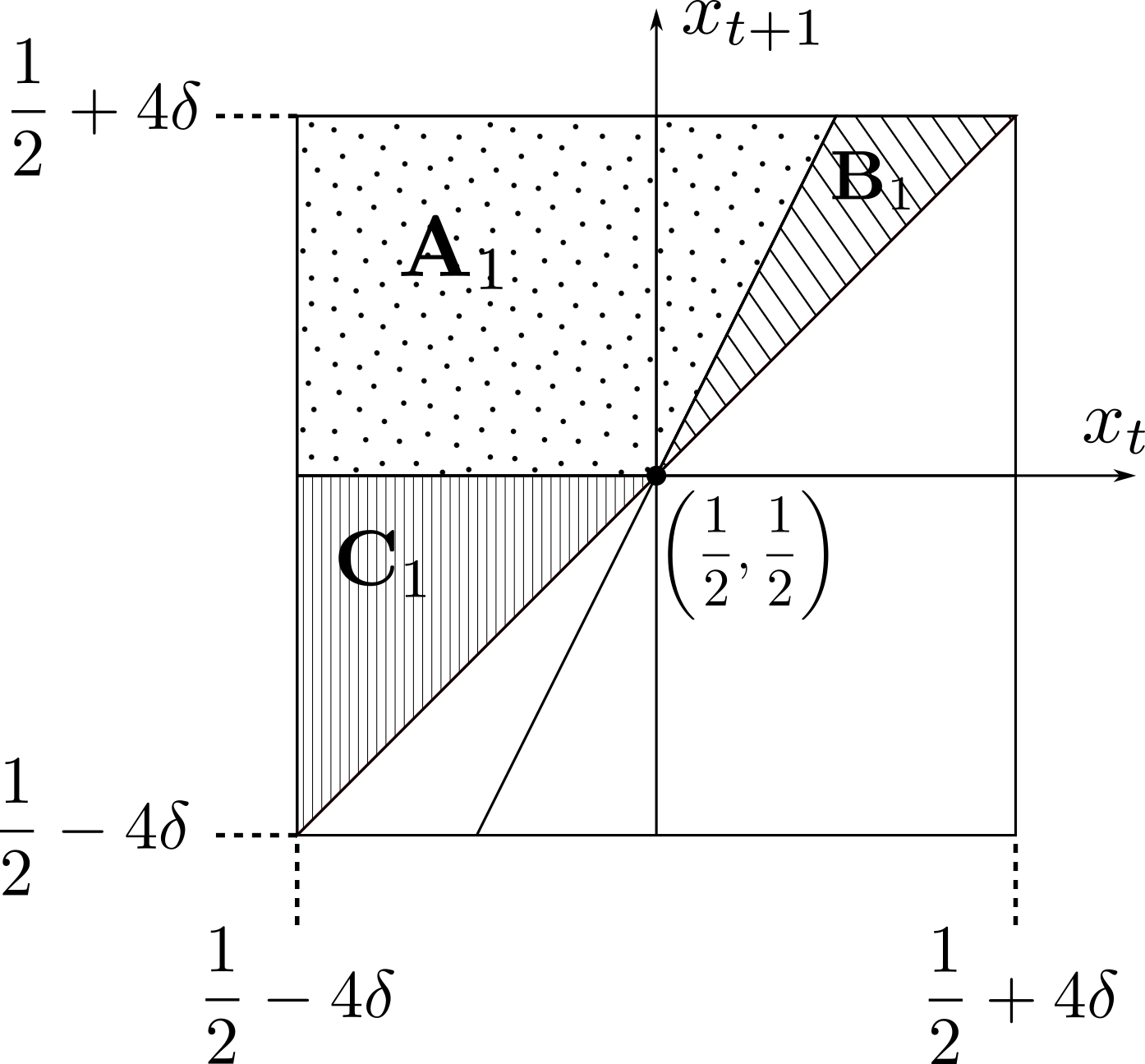}
    \caption{Partitioning the ~$\yellow'$ domain.}
    \label{fig:partition_yellow}
\end{figure}

\begin{remark}
    At various times throughout our analysis, we would like to calculate different statistical properties of the system at round $t+2$, conditioning on $(\xx{t},\xx{t+1})\in \calG$, as was done in e.g., Observation~\ref{lem:evolution}. 
    For the sake of clarity of presentation, in all subsequent cases, we shall 
    omit the conditioning notation.
    The reader should therefore keep in mind, that whenever such properties are calculated, they are actually done while conditioning on $\x_t = \xx{t}$ and $\x_{t+1} = \xx{t+1}$, where the point $(\xx{t},\xx{t+1})$ would always be clear from the context. For example, in Lemma~\ref{lem:A2} below, the probability $\bbP\pa{ (\x_{t+1},\x_{t+2}) \notin \yellow'}$ is calculated conditioning on $\x_t = \xx{t},~ \x_{t+1} = \xx{t+1}$, where 
    $(\xx{t},\xx{t+1})\in \bfA$, yet this conditioning is omitted.
\end{remark}

In the next lemma, we study the distribution of the future location of any point~$(\x_t,\x_{t+1})\in\bfA$. This area happens to be ideal to escape~$\yellow'$, because it allows the Markov chain to quickly build up ``speed''.
Item (a) in the next lemma says that, with some probability that depends on the current speed the following occur: (1) the speed in the following round increases by a factor of two, and (2) the process in the next round either remains in $\bfA$, or goes outside of ~$\yellow'$. Note that when the current speed is not too low, that is, when $\x_{t+1}-\x_t > 1/\sqrt{n}$, this combined event happens with constant probability. Item (b) says that with constant probability, (1) the speed in the next round would not be too low, and (2) the process either remains in $\bfA$, or goes outside of ~$\yellow'$.

\begin{lemma}\label{lem:A2}
    If~$(\x_t,\x_{t+1}) \in \bfA$, and provided that~$\delta$ is small enough and~$n$ is large enough,
    \begin{itemize}
        \item[(a)] We have
           $ \bbP\pa{ (\x_{t+1},\x_{t+2}) \notin \yellow' \setminus \bfA  \medcap |\x_{t+2} - \x_{t+1}| > 2 |\x_{t+1} - \x_t| } > 1-\exp \pa{- 3 n \cdot (\x_{t+1} - \x_t)^2 }$.
        \item[(b)] There exists a constant~$c_2 = c_2(c)>0$ s.t.
            $\bbP\pa{ (\x_{t+1},\x_{t+2}) \notin \yellow' \setminus \bfA \medcap |\x_{t+2}-\x_{t+1}| > 1/ \sqrt{n} } > c_2$.
    \end{itemize}
\end{lemma}

Now, we can iteratively use the previous result to prove that any state in~$\bfA$ has a reasonable chance to escape~$\yellow'$. The proofs of both Lemmas~\ref{lem:A2} and~\ref{lem:A3} are deferred to Appendix~\ref{app:yellow_A}.
\begin{lemma} \label{lem:A3}
    There is a constant~$c_3 = c_3(c)$ s.t.~if~$(\x_{t_0},\x_{t_0+1}) \in \bfA$, then 
        $\bbP \pa{\exists t_1 < t_0 + \log n, (\x_{t_1},\x_{t_1+1}) \notin \yellow'} > c_3$.
\end{lemma}
We are left with proving that the system cannot be stuck in~$\bfB$ or~$\bfC$ for too long. We start with~$\bfB$.
The analysis of this area is relatively complex, because it is difficult to rule out the possibility that the Markov chain remains there at a low speed. 
We prove that any state in~$\bfB$ must either make a step towards escaping~$\yellow'$, or have a good chance of leaving~$\bfB$. The proof of the following lemma is given in Section~\ref{sec:lem:B1}.
\begin{lemma} \label{lem:B1}
    There are constants~$c_4,c_5>0$ such that
    if~$(\x_{t},\x_{t+1}) \in \bfB$, then either
    \begin{itemize}
        \item[(a)] $|\x_{t+1}-1/2| > \pa{1+c_4/\sqrt{\ell}} |\x_{t}-1/2|$,
        or
        \item[(b)] $\bbP \pa{(\x_{t+1},\x_{t+2}) \notin \bfB} > c_5$.
    \end{itemize}
\end{lemma}

Now, we can iteratively use the previous result to prove that any state in~$\bfB$ either leaves~$\bfB$ or escapes~\yellow' in a reasonable amount of time. The proof of Lemma~\ref{lem:B2} is deferred to Section~\ref{sec:yellow_B}.
\begin{lemma} \label{lem:B2}
    If~$(\x_{t_0},\x_{t_0+1}) \in \bfB$, then, w.h.p.,
        $\min \{ t \geq t_0, (\x_{t},\x_{t_1}) \notin \bfB \} < t_0 + \frac{\sqrt{c}}{c_4} \cdot \log^{3/2} n$.
\end{lemma}

We are left with proving that the system cannot stay in~$\bfC$ for too long. 
Fortunately, from this area, the Markov chain is naturally pushed towards~$\bfA$, which makes the analysis simple. The proof of Lemma~\ref{lem:C} is deferred to Section~\ref{sec:yellow_C}.

\begin{lemma} \label{lem:C}
    There is a constant~$c_6>0$ such that if~$(\x_{t},\x_{t+1}) \in \bfC$, then
    \begin{equation*}
        \max \big\{ \bbP \pa{(\x_{t+1},\x_{t+2}) \notin \yellow'\setminus \bfA} , \bbP \pa{(\x_{t+2},\x_{t+3}) \notin \yellow'\setminus\bfA} \big\} > c_6.
    \end{equation*}
\end{lemma}
Eventually, we have all the necessary results to conclude the proof regarding the Yellow area.
\begin{proof} [Proof of Lemma~\ref{lem:yellow2}]
    By Lemma~\ref{lem:A3}, if~$(\x_{t_0},\x_{t_0+1}) \in \bfA$, then $\bbP \pa{\exists t_1 < t_0 + \log n, (\x_{t_1},\x_{t_1+1}) \notin \yellow'} > c_3 > 0$.
    By Lemma~\ref{lem:C}, this implies that if~$(\x_{t_0},\x_{t_0+1}) \in \bfA \cup \bfC$,
    \begin{equation} \label{eq:escape}
        \bbP \pa{\exists t_1 < t_0 + \log n +2, (\x_{t_1},\x_{t_1+1}) \notin \yellow'} > \min (c_3, c_3 \cdot c_6) =  c_3 \cdot c_6 > 0.
    \end{equation}
    By Lemma~\ref{lem:B2}, w.h.p., whenever the process is at~$\bfB$, it does not spend more than~$(\sqrt{c} / c_4)\cdot \log^{3/2} n$ consecutive rounds there. This means, that for any constant $c'>0$,  during~$c'\log^{5/2} n$ consecutive rounds, w.h.p., we must either leave~$\yellow'$ or be at~$\bfA \cup \bfC$ on at least~$\frac{c'\log^{5/2} n}{(\sqrt{c} / c_4)\cdot \log^{3/2} n} = \frac{c' c_4}{\sqrt{c}} \cdot \log n$ 
    distinct rounds. By Eq.~\eqref{eq:escape}, the probability that the system fails to escape~$\yellow'$ in each of these occasions is at most~$(1-c_3 \cdot c_6)^{(c' c_4/\sqrt{c}) \cdot \log n}$. 
    Taking $c'$ to be sufficiently large
     concludes the proof of Lemma~\ref{lem:yellow2}.
\end{proof}



\subsection{Proof of Lemma~\ref{lem:B1} --- The main lemma regarding area B}\label{sec:lem:B1}

The goal of this section is to prove Lemma~\ref{lem:B1}, which concerns Area B inside the Yellow domain. Without loss of generality, we may assume that~$(\x_t,\x_{t+1}) \in \bfB_1$ (the same arguments apply to~$\bfB_0$ symmetrically).
    
The proof shall use the following lemma which concerns competition between two coins, both being tossed $k$ times. One coin has a probability of $p$ to be ``heads'', and the other has probability $q$. We assume that $p<q$, although the difference is very small, $q-p \leq 1/ \sqrt{k}$. The lemma upper bounds the probability that the more likely coin wins, that is, that the coin with parameter $q$ falls on ``heads'' more times than the other coin. The proof is deferred to Appendix~\ref{app:underdog}.
    
\begin{lemma} \label{lem:handcrafted_converse}
    There exists a constant~$\alpha>1$, s.t. for every integer~$k$, every~$p,q \in [1/3,2/3]$ with~$p<q$ and~$q-p \leq 1/ \sqrt{k}$,
    we have
    \begin{equation*}
        \bbP \pa{B_k(p) < B_k(q)} < \frac{1}{2} + \alpha (q-p) \sqrt{k} - \frac{1}{2}\bbP\pa{ B_k(p) = B_k(q)}.
    \end{equation*}
\end{lemma}
    
Next, let us define
\begin{equation} \label{eq:def_g}
    g(x,y) = \prob{\ell}{y}{>}{\ell}{x} +y\cdot \prob{\ell}{y}{=}{\ell}{x} + \frac{1}{n} \pa{1-\prob{\ell}{y}{\geq}{\ell}{x} },
\end{equation}
so that, conditioning on~$(\x_t,\x_{t+1})$, $\bbE(\x_{t+2}) = g(\x_t,\x_{t+1})$ by Observation~\ref{lem:evolution}.
We start with the following claim, which will be used to prove the subsequent claim. The proof of Claim~\ref{claim:last_minute_obstacle} is deferred to Appendix~\ref{app:last_minute_obstacle}. 
\begin{claim} \label{claim:last_minute_obstacle}
    Let~$x \in [1/3,2/3]$. 
    On the interval~$[x,x+1 / \sqrt{\ell}]$, and for~$\ell$ large enough, $y \mapsto g(x,y)-y$ is a strictly increasing function of~$y$.
\end{claim}
The next claim concerns the fixed points of $g(x,y)$ as a function of~$y$.
\begin{claim} \label{claim:def_f}
    For any given~$x \in [1/2+4/n,1/2+4\delta]$,
    as a function of~$y$, the equation~$y = g(x,y)$ has at most one solution on the interval~$[x,x+1/ \sqrt{\ell}]$. Moreover, in the case that it has no solution, then~$g(x,x+1/\sqrt{\ell}) < x+1/ \sqrt{\ell}$.
\end{claim}
\begin{proof}
    First, we claim that~$g(x,x) < x$. 
    Let~$p = \bbP(B_\ell(x) > B_\ell(x))$ and~$q = \bbP(B_\ell(x) = B_\ell(x))$. We rearrange the definition of~$g$ slightly to obtain~$g(x,x) = p + x\cdot q + \frac{p}{n}$.
    Moreover, $x = x \cdot (2 p + q) \geq \pa{1+8/n} \cdot p + x\cdot q > g(x,x)$,
    where the first inequality is because~$x \geq 1/2+4/n$.   Next, let~$h(y) = g(x,y)-y$. Function~$h$ is continuous, and what we just showed implies~$h(x) < 0$.
    Moreover, by Claim~\ref{claim:last_minute_obstacle}, we know that~$h$ is strictly increasing on~$[x,x+1/\sqrt{\ell}]$.
    Therefore, either~$h(x+1/ \sqrt{\ell}) \geq 0$, in which case there is a unique~$y^\star \in [x,x+1/\sqrt{\ell}]$ such that~$h(y^\star) = 0$;
    or~$h(x+1/ \sqrt{\ell}) < 0$, i.e., $g(x,x+1/\sqrt{\ell}) < x+1/ \sqrt{\ell}$, in which case the equation~$y = g(x,y)$ has no solution on the interval.
\end{proof}
For every~$x \in [1/2+4/n,1/2+4\delta]$, let~$f(x)$ be the solution of the equation~$y = g(x,y)$ in the interval~$[x,x+1/\sqrt{\ell}]$ if it exists, and~$f(x) = x+1/\sqrt{\ell}$ otherwise. Note that by Claim~\ref{claim:def_f}, with this definition we always have
\begin{equation} \label{eq:f_prop}
    g(x,f(x)) \leq f(x).
\end{equation}

\begin{claim} \label{claim:for_clarity}
    For any~$x \in [1/2+4/n,1/2+4\delta]$, it holds that
    \begin{equation*}
        f(x)-x > \frac{1}{2\alpha \sqrt{\ell}} \pa{x-\frac{1}{2}},
    \end{equation*}
    where~$\alpha > 1$ is the constant stated in Lemma~\ref{lem:handcrafted_converse}.
\end{claim}
\begin{proof}
    If~$f(x)$ is not a solution to~$y = g(x,y)$, then by definition~$f(x) = x+1/ \sqrt{\ell}$, i.e.,
    \begin{equation*}
        f(x) - x = \frac{1}{\sqrt{\ell}} > \frac{1}{2\alpha \sqrt{\ell}} \pa{x-\frac{1}{2}},
    \end{equation*}
    and so the statement holds. Otherwise, then~$f(x) = g(x,f(x))$ and belongs to~$[x,x+1/\sqrt{\ell}]$.
    By Lemma~\ref{lem:handcrafted_converse}, there exists~$\alpha>0$ s.t.
    \begin{equation*}
        \prob{\ell}{f(x)}{>}{\ell}{x} < \frac{1}{2} + \alpha (f(x)-x) \sqrt{\ell} - \frac{1}{2}\prob{\ell}{f(x)}{=}{\ell}{x}.
    \end{equation*}
    This can be plugged into the definition of~$f$ (Eq.~\eqref{eq:def_g}) to give
    \begin{equation*}
        f(x) < \frac{1}{2} + \alpha (f(x)-x) \sqrt{\ell} + \pa{ f(x) - \frac{1}{2}} \prob{\ell}{f(x)}{=}{\ell}{x} +\frac{1}{n}
    \end{equation*}
    which we can rewrite,
    \begin{equation*}
        \pa{1-\prob{\ell}{f(x)}{=}{\ell}{x}} \pa{f(x)-\frac{1}{2}} < \alpha (f(x)-x) \sqrt{\ell} + \frac{1}{n}.
    \end{equation*}
    This gives
    \begin{equation*}
        f(x)-x > \frac{1-\prob{\ell}{f(x)}{=}{\ell}{x}}{\alpha \sqrt{\ell}} \pa{f(x)-\frac{1}{2}} -\frac{1}{\alpha \cdot n \sqrt{\ell}} > \frac{1}{2\alpha \sqrt{\ell}} \pa{x-\frac{1}{2}-\frac{2}{n}},
    \end{equation*}
    where the last inequality comes from the upper bound~$\bbP\pa{ B_\ell(f(x)) = B_\ell(x)} < 1/2$ (which is true when~$\ell$ is large enough), and from the fact that~$f(x)>x$.
    Since~$(x-1/2) \geq 4/n$, this implies
    \begin{equation*}
        f(x)-x > \frac{1}{4\alpha \sqrt{\ell}} \pa{x-\frac{1}{2}},
    \end{equation*}
    as desired. This completes the proof of Claim~\ref{claim:for_clarity}.
\end{proof}
    
Next, rewriting~$f(x)-x = \pa{f(x)-1/2} - \pa{x-1/2}$, we get from Claim~\ref{claim:for_clarity} that for every~$x \in [1/2+4/n,1/2+4\delta]$,
\begin{equation} \label{eq:aux_B1}
    \pa{f(x)-\frac{1}{2}} > \pa{ 1 + \frac{1}{4\alpha\sqrt{\ell}} } \cdot \pa{x-\frac{1}{2}}.
\end{equation}
We are now ready to  conclude the proof of Lemma~\ref{lem:B1}. Let~$c_4 = 1/4\alpha$.
\begin{itemize}
    \item If~$\x_t \in [1/2,1/2+4/n]$, then by definition of~$\bfB$, $\x_{t+1} \in [1/2,1/2+8/n]$. For the same reason, for~$(\x_{t+1},\x_{t+2})$ to be in~$\bfB$, it is necessary that~$\x_{t+2} \in [1/2,1/2+16/n]$. By Lemma~\ref{lem:noise} (see Appendix~ref{sec:noise}), there is a constant probability that it is not the case, and so~{\em (b)} (in the statement of the lemma) holds.
    \item Otherwise, if~$\x_t \in [1/2+4/n,1/2+4\delta]$ and~$\x_{t+1} > f(\x_t)$, then by Eq.~\eqref{eq:aux_B1},
    \begin{equation*}
        \x_{t+1}-\frac{1}{2} >f(\x_t)-\frac{1}{2} > \pa{1+\frac{c_4}{\sqrt{\ell}}} \pa{\x_{t}-\frac{1}{2}},
    \end{equation*}
    and so~{\em (a)} holds.
    \item Else, $f(\x_t) \geq \x_{t+1}$. Moreover, by the definitions of~$f$ and~$\bfB_1$, we have the following relation:
    \begin{equation}\label{eq:extra_clarity}
        \x_t+1/ \sqrt{\ell} \geq f(\x_t) \geq \x_{t+1} \geq \x_t.
    \end{equation}
    By Eq.~\eqref{eq:f_prop}, $g(\x_t,f(x_t)) \leq f(\x_t)$, i.e., $g(\x_t,f(x_t)) - f(\x_t) \leq 0$.
    By Claim~\ref{claim:last_minute_obstacle}, for~$\ell$ large enough, function $y \mapsto g(\x_t,y) - y$ is strictly increasing on~$[\x_t,\x_t+1/ \sqrt{\ell}]$.
    Eq.~\eqref{eq:extra_clarity} ensures that~$\x_{t+1}$ and~$f(\x_t)$ are within this interval, so $g(\x_t,\x_{t+1}) - \x_{t+1} \leq g(\x_t,f(x_t)) - f(\x_t) \leq 0$, i.e., $g(\x_t,\x_{t+1}) \leq \x_{t+1}$.
    Recall that~$\bbE(\x_{t+2}) = g(\x_t,\x_{t+1})$ -- therefore, $\bbE(\x_{t+2}) \leq \x_{t+1}$.
    By Lemma~\ref{lem:anti_concentration} (see Appendix \ref{sec:noise}), there is a constant probability~$c_5$ that~$\x_{t+2} < \x_{t+1}$. 
    If this the case, since~$x_{t+1} > 1/2$ and by the definition of~$\bfB$, we get that~$(x_{t+1},x_{t+2}) \notin \bfB$ and~{\em (b)} holds.
\end{itemize}
This concludes the proof of Lemma~\ref{lem:B1}. \qed

\section{Proof of Lemma~\ref{lem:cyan} --- Cyan area} \label{sec:lem:cyan}

The goal of this section is to prove Lemma~\ref{lem:cyan}, which concerns Area~$\cyan{}$.
We distinguish between two cases.
    
{\em Case 1.}
$\x_{t_0} \geq 1/\log(n)$.
In this case, by definition of~$\cyan{1}$, we must have~$\x_{t_0+1} < 1/ \log(n)$. Note that in this case,
for~$n$ large enough, $\x_{t_0+1}-\delta < 1/\log(n) - \delta < 0$.
Then,
\begin{itemize}
    \item either~$\x_{t_0+2} < \x_{t_0+1}+\delta$. In this case, $\x_{t_0+1}-\delta < 0 < \x_{t_0+2} < \x_{t_0+1}+\delta$, and so~$(\x_{t_0+1},\x_{t_0+2}) \in \cyan{1}$ (but this time Case 2 applies). 
    \item or~$\x_{t_0+2} \geq \x_{t_0+1}+\delta$, and so~$(\x_{t_0+1},\x_{t_0+2}) \in \green{1}$.
    \item (We can't have~$\x_{t_0+2} = 0$ because the source is assumed to have opinion~$1$.)
\end{itemize}
    
{\em Case 2.}
$\x_{t_0} < 1/\log(n)$.
Let~$\gamma = \gamma(c) = (1-1/e) \cdot \exp(-2c)/2$ and let~$K = K(c) = c \cdot \exp\pa{-2c}/2$. We will study separately three ranges of value for~$\x_{t+1}$.
Claim~\ref{claim:cyan_small} below concerns small values of~$\x_{t+1}$, Claim~\ref{claim:cyan_intermediate} concerns intermediate values of~$\x_{t+1}$, and Claim~\ref{claim:cyan_large} concerns large values of~$\x_{t+1}$. Their proofs follow from simple applications of Chernoff's bound, and so they are deferred to Appendix~\ref{app:cyan}. 
    
\begin{claim}  \label{claim:cyan_small}
    If~$\x_t < 1/ \log(n)$, and if~$0 < \x_{t+1} \leq 1/ \ell$, then
        $\bbP \pa{ \x_{t+2} > \frac{K}{2} \x_{t+1} \log n } > 1-\exp \pa{-\frac{K}{8} \log n}$.
\end{claim}
    
\begin{claim}  \label{claim:cyan_intermediate} 
    If~$\x_t < 1/ \log(n)$, and if~$1/\ell < \x_{t+1} \leq \gamma$, then
        $\bbP \pa{ \x_{t+2} > \gamma } > 1-\exp \pa{-\frac{\gamma n}{8}} > 1-\exp \pa{-\frac{K}{8} \log n}$.
\end{claim}
    
\begin{claim} \label{claim:cyan_large}
    If~$\x_t < 1/ \log(n)$, and if~$\x_{t+1} > \gamma$, then
        $\bbP \pa{ \x_{t+2} > \frac{1}{2} } > 1-\exp \pa{-\frac{n}{18}} > 1-\exp \pa{-\frac{K}{8} \log n}$.
\end{claim}
We say that a round~$t$ is {\em successful} if~$(\x_{t},\x_{t+1}) \in \cyan{1}$, and the event of either Claim~\ref{claim:cyan_small}, \ref{claim:cyan_intermediate} or~\ref{claim:cyan_large} happens. Formally, 
\begin{equation*}
\begin{cases}
    \x_{t+1} \leq 1/ \ell \text{ and } \x_{t+2} > K \x_{t+1} \log n \text{ (corresponding to Claim~\ref{claim:cyan_small}), or} \\
    1/ \ell <~ \x_{t+1} \leq \gamma \text{ and } \x_{t+2} > \gamma \text{ (corresponding to Claim~\ref{claim:cyan_intermediate}), or} \\
    \gamma <~ \x_{t+1} \text{ and } (\x_{t+1},\x_{t+2}) \notin \cyan{1} \text{ (corresponding to Claim~\ref{claim:cyan_large}).}
\end{cases}
\end{equation*}
Let~$X$ be the number of successful rounds starting from~$t_0$.
This definition implies that necessarily,
\begin{equation*}
    X < \frac{\log(n / \ell)}{\log(K \cdot \log n/2)}+2 := X_{\max}.
\end{equation*}
Indeed, since~$\x_{t_0+1} > 1/n$ (by definition of~$\cyan{1}$),
\begin{itemize}
    \item $\log(n / \ell) / \log(K \cdot \log(n)/2)$ rounds are always enough to get~$\x_{t+1} > 1/ \ell$
    \item one more round is enough to get~$\x_{t+1} > \gamma$
    \item one more round is enough to get~$\x_{t+1} > 1/2$, in which case~$(\x_t,\x_{t+1}) \notin \cyan{1}$.
\end{itemize}
Therefore, by Claims~\ref{claim:cyan_small}, \ref{claim:cyan_intermediate} and~\ref{claim:cyan_large}, the probability that, starting from~$t_0$, all rounds are successful until the system is out of~$\cyan{1}$ is at least
$\pa{ 1-\exp \pa{-\frac{K}{8} \log n} }^{X_{\max}} \geq
    1 - X_{\max} \cdot \exp \pa{-\frac{K}{8} \log n} = 1-1/n^{\Omega(1)}$.
Moreover, for any successful round~$t$, $\x_{t+2} > \x_{t+1}$ (by definition of a successful round) and~$\x_{t+1} < \delta + 1/ \log(n)$ (this is a straightforward consequence of the definition of~$\cyan{1}$). Thus, by construction of the partition, we must have~$(\x_{t+1},\x_{t+2}) \in \cyan{1} \cup \green{1} \cup \purple{1}$.
This implies that~$(\x_{t_1},\x_{t_1+1}) \in \green{1} \cup \purple{1}$, which concludes the proof of Lemma~\ref{lem:cyan}.

\section{Discussion and Future Work}
This paper considers a natural problem of information spreading in a self-stabilizing context, where it is assumed that a source agent has useful knowledge about the environment, and others would like to learn this information without being able to distinguish the source from non-source agents. Motivated by biological scenarios, our focus is on solutions that utilize passive communication. We identify an extremely simple algorithm, called FET (Protocol \ref{protocol:FET}), which has a natural appeal: In each round, each (non-source) agent estimates the current tendency direction of the dynamics, and then adapts to the emerging trend. The correct operation of the algorithm does not require that the source actively cooperates with the algorithm, and instead, only assumes that it maintains its correct option throughout the execution. 

Different performance parameters may be further optimized in future work. For example, our analysis uses $O(\log n)$ samples per round, and it would be interesting to see whether the problem can be solved in poly-logarithmic time w.h.p, by using only a constant number of samples per round. Also, we do not exclude the possibility that a tighter analysis of Algorithm FET would reduce our bound on the running time. In addition, our framework assumes the presence of a single source agent, but as mentioned, it can also allow for a constant number of sources, as long as it is guaranteed that all sources agree on the correct opinion. No attempt has been made to consider a larger regime of sources (beyond a constant), although we believe that such a framework is also manageable. 

Finally, as a more philosophical remark, we note that early adapting to emerging trends is a common strategy in humans, which is, in some sense, encouraged by modern economic systems. For example, investing in a successful company can yield large revenues, especially if such an investment is made before others notice its high potential. 
On a global scale, the collective benefits of this strategy are typically associated with economic growth. This paper shows that such a strategy can also have a collective benefit that traces back to basic aspects of collective decision-making, suggesting the possibility that it may have evolved via group-selection.  With this in mind, it would be interesting to empirically check whether such a strategy exists also in other animal groups, e.g., fish schools \cite{sumpter2008consensus} or ants \cite{rajendran2022ants}.

\section*{Acknowledgements}

The authors would particularly like to thank Pierre Fraigniaud and Ofer Feinerman for helpful and exciting discussions. In addition, we would like to thank Uriel Feige, Uri Zwick, Bernard Haeupler, and Emanuele Natale for preliminary discussions.

\bibliographystyle{ACM-Reference-Format}
\bibliography{references}


\begin{thebibliography}{37}


\ifx \showCODEN    \undefined \def \showCODEN     #1{\unskip}     \fi
\ifx \showDOI      \undefined \def \showDOI       #1{#1}\fi
\ifx \showISBNx    \undefined \def \showISBNx     #1{\unskip}     \fi
\ifx \showISBNxiii \undefined \def \showISBNxiii  #1{\unskip}     \fi
\ifx \showISSN     \undefined \def \showISSN      #1{\unskip}     \fi
\ifx \showLCCN     \undefined \def \showLCCN      #1{\unskip}     \fi
\ifx \shownote     \undefined \def \shownote      #1{#1}          \fi
\ifx \showarticletitle \undefined \def \showarticletitle #1{#1}   \fi
\ifx \showURL      \undefined \def \showURL       {\relax}        \fi
\providecommand\bibfield[2]{#2}
\providecommand\bibinfo[2]{#2}
\providecommand\natexlab[1]{#1}
\providecommand\showeprint[2][]{arXiv:#2}

\bibitem[\protect\citeauthoryear{Alistarh, Aspnes, Eisenstat, Gelashvili, and
  Rivest}{Alistarh et~al\mbox{.}}{2016}]%
        {DBLP:journals/corr/AlistarhAEGR16}
\bibfield{author}{\bibinfo{person}{Dan Alistarh}, \bibinfo{person}{James
  Aspnes}, \bibinfo{person}{David Eisenstat}, \bibinfo{person}{Rati
  Gelashvili}, {and} \bibinfo{person}{Ronald~L. Rivest}.}
  \bibinfo{year}{2016}\natexlab{}.
\newblock \showarticletitle{Time-Space Trade-offs in Population Protocols}.
\newblock \bibinfo{journal}{\emph{CoRR}}  \bibinfo{volume}{abs/1602.08032}
  (\bibinfo{year}{2016}).
\newblock
\showeprint[arXiv]{1602.08032}
\urldef\tempurl%
\url{http://arxiv.org/abs/1602.08032}
\showURL{%
\tempurl}


\bibitem[\protect\citeauthoryear{Alistarh and Gelashvili}{Alistarh and
  Gelashvili}{2015}]%
        {DBLP:conf/icalp/AlistarhG15}
\bibfield{author}{\bibinfo{person}{Dan Alistarh} {and} \bibinfo{person}{Rati
  Gelashvili}.} \bibinfo{year}{2015}\natexlab{}.
\newblock \showarticletitle{Polylogarithmic-Time Leader Election in Population
  Protocols}. In \bibinfo{booktitle}{\emph{Automata, Languages, and Programming
  - 42nd International Colloquium, {ICALP} 2015, Kyoto, Japan, July 6-10, 2015,
  Proceedings, Part {II}}} \emph{(\bibinfo{series}{Lecture Notes in Computer
  Science}, Vol.~\bibinfo{volume}{9135})},
  \bibfield{editor}{\bibinfo{person}{Magn{\'{u}}s~M. Halld{\'{o}}rsson},
  \bibinfo{person}{Kazuo Iwama}, \bibinfo{person}{Naoki Kobayashi}, {and}
  \bibinfo{person}{Bettina Speckmann}} (Eds.). \bibinfo{publisher}{Springer},
  \bibinfo{pages}{479--491}.
\newblock
\urldef\tempurl%
\url{https://doi.org/10.1007/978-3-662-47666-6\_38}
\showDOI{\tempurl}


\bibitem[\protect\citeauthoryear{Alistarh and Gelashvili}{Alistarh and
  Gelashvili}{2018}]%
        {DBLP:journals/sigact/AlistarhG18}
\bibfield{author}{\bibinfo{person}{Dan Alistarh} {and} \bibinfo{person}{Rati
  Gelashvili}.} \bibinfo{year}{2018}\natexlab{}.
\newblock \showarticletitle{Recent Algorithmic Advances in Population
  Protocols}.
\newblock \bibinfo{journal}{\emph{{SIGACT} News}} \bibinfo{volume}{49},
  \bibinfo{number}{3} (\bibinfo{year}{2018}), \bibinfo{pages}{63--73}.
\newblock
\urldef\tempurl%
\url{https://doi.org/10.1145/3289137.3289150}
\showDOI{\tempurl}


\bibitem[\protect\citeauthoryear{Angluin, Aspnes, Diamadi, Fischer, and
  Peralta}{Angluin et~al\mbox{.}}{2006}]%
        {angluin2006computation}
\bibfield{author}{\bibinfo{person}{Dana Angluin}, \bibinfo{person}{James
  Aspnes}, \bibinfo{person}{Zo{\"e} Diamadi}, \bibinfo{person}{Michael~J
  Fischer}, {and} \bibinfo{person}{Ren{\'e} Peralta}.}
  \bibinfo{year}{2006}\natexlab{}.
\newblock \showarticletitle{Computation in networks of passively mobile
  finite-state sensors}.
\newblock \bibinfo{journal}{\emph{Distributed computing}} \bibinfo{volume}{18},
  \bibinfo{number}{4} (\bibinfo{year}{2006}), \bibinfo{pages}{235--253}.
\newblock


\bibitem[\protect\citeauthoryear{Angluin, Aspnes, and Eisenstat}{Angluin
  et~al\mbox{.}}{2008}]%
        {DBLP:journals/dc/AngluinAE08}
\bibfield{author}{\bibinfo{person}{Dana Angluin}, \bibinfo{person}{James
  Aspnes}, {and} \bibinfo{person}{David Eisenstat}.}
  \bibinfo{year}{2008}\natexlab{}.
\newblock \showarticletitle{A simple population protocol for fast robust
  approximate majority}.
\newblock \bibinfo{journal}{\emph{Distributed Comput.}} \bibinfo{volume}{21},
  \bibinfo{number}{2} (\bibinfo{year}{2008}), \bibinfo{pages}{87--102}.
\newblock
\urldef\tempurl%
\url{https://doi.org/10.1007/s00446-008-0059-z}
\showDOI{\tempurl}


\bibitem[\protect\citeauthoryear{Aspnes and Ruppert}{Aspnes and
  Ruppert}{2007}]%
        {DBLP:journals/eatcs/AspnesR07}
\bibfield{author}{\bibinfo{person}{James Aspnes} {and} \bibinfo{person}{Eric
  Ruppert}.} \bibinfo{year}{2007}\natexlab{}.
\newblock \showarticletitle{An Introduction to Population Protocols}.
\newblock \bibinfo{journal}{\emph{Bull. {EATCS}}}  \bibinfo{volume}{93}
  (\bibinfo{year}{2007}), \bibinfo{pages}{98--117}.
\newblock


\bibitem[\protect\citeauthoryear{Ayalon, Sternklar, Fonio, Korman, Gov, and
  Feinerman}{Ayalon et~al\mbox{.}}{2021}]%
        {DBLP:journals/fams/AyalonSFKGF21}
\bibfield{author}{\bibinfo{person}{Oran Ayalon}, \bibinfo{person}{Yigal
  Sternklar}, \bibinfo{person}{Ehud Fonio}, \bibinfo{person}{Amos Korman},
  \bibinfo{person}{Nir~S. Gov}, {and} \bibinfo{person}{Ofer Feinerman}.}
  \bibinfo{year}{2021}\natexlab{}.
\newblock \showarticletitle{Sequential Decision-Making in Ants and Implications
  to the Evidence Accumulation Decision Model}.
\newblock \bibinfo{journal}{\emph{Frontiers Appl. Math. Stat.}}
  \bibinfo{volume}{7} (\bibinfo{year}{2021}), \bibinfo{pages}{672773}.
\newblock
\urldef\tempurl%
\url{https://doi.org/10.3389/fams.2021.672773}
\showDOI{\tempurl}


\bibitem[\protect\citeauthoryear{Barclay}{Barclay}{1982}]%
        {barclay1982interindividual}
\bibfield{author}{\bibinfo{person}{Robert~MR Barclay}.}
  \bibinfo{year}{1982}\natexlab{}.
\newblock \showarticletitle{Interindividual use of echolocation calls:
  eavesdropping by bats}.
\newblock \bibinfo{journal}{\emph{Behavioral Ecology and Sociobiology}}
  \bibinfo{volume}{10}, \bibinfo{number}{4} (\bibinfo{year}{1982}),
  \bibinfo{pages}{271--275}.
\newblock


\bibitem[\protect\citeauthoryear{Bastide, Giakkoupis, and Saribekyan}{Bastide
  et~al\mbox{.}}{2021}]%
        {bastide2021self}
\bibfield{author}{\bibinfo{person}{Paul Bastide}, \bibinfo{person}{George
  Giakkoupis}, {and} \bibinfo{person}{Hayk Saribekyan}.}
  \bibinfo{year}{2021}\natexlab{}.
\newblock \showarticletitle{Self-Stabilizing Clock Synchronization with 1-bit
  Messages}. In \bibinfo{booktitle}{\emph{Proceedings of the 2021 ACM-SIAM
  Symposium on Discrete Algorithms (SODA)}}. SIAM, \bibinfo{pages}{2154--2173}.
\newblock


\bibitem[\protect\citeauthoryear{Becchetti, Clementi, and Natale}{Becchetti
  et~al\mbox{.}}{2020}]%
        {becchetti2020consensus}
\bibfield{author}{\bibinfo{person}{Luca Becchetti}, \bibinfo{person}{Andrea
  Clementi}, {and} \bibinfo{person}{Emanuele Natale}.}
  \bibinfo{year}{2020}\natexlab{}.
\newblock \showarticletitle{Consensus dynamics: An overview}.
\newblock \bibinfo{journal}{\emph{ACM SIGACT News}} \bibinfo{volume}{51},
  \bibinfo{number}{1} (\bibinfo{year}{2020}), \bibinfo{pages}{58--104}.
\newblock


\bibitem[\protect\citeauthoryear{Boczkowski, Feinerman, Korman, and
  Natale}{Boczkowski et~al\mbox{.}}{2018a}]%
        {DBLP:conf/innovations/BoczkowskiFKN18}
\bibfield{author}{\bibinfo{person}{Lucas Boczkowski}, \bibinfo{person}{Ofer
  Feinerman}, \bibinfo{person}{Amos Korman}, {and} \bibinfo{person}{Emanuele
  Natale}.} \bibinfo{year}{2018}\natexlab{a}.
\newblock \showarticletitle{Limits for Rumor Spreading in Stochastic
  Populations}. In \bibinfo{booktitle}{\emph{9th Innovations in Theoretical
  Computer Science Conference, {ITCS} 2018, January 11-14, 2018, Cambridge, MA,
  {USA}}} \emph{(\bibinfo{series}{LIPIcs}, Vol.~\bibinfo{volume}{94})},
  \bibfield{editor}{\bibinfo{person}{Anna~R. Karlin}} (Ed.).
  \bibinfo{publisher}{Schloss Dagstuhl - Leibniz-Zentrum f{\"{u}}r Informatik},
  \bibinfo{pages}{49:1--49:21}.
\newblock
\urldef\tempurl%
\url{https://doi.org/10.4230/LIPIcs.ITCS.2018.49}
\showDOI{\tempurl}


\bibitem[\protect\citeauthoryear{Boczkowski, Korman, and Natale}{Boczkowski
  et~al\mbox{.}}{2019}]%
        {DBLP:journals/dc/BoczkowskiKN19}
\bibfield{author}{\bibinfo{person}{Lucas Boczkowski}, \bibinfo{person}{Amos
  Korman}, {and} \bibinfo{person}{Emanuele Natale}.}
  \bibinfo{year}{2019}\natexlab{}.
\newblock \showarticletitle{Minimizing message size in stochastic communication
  patterns: fast self-stabilizing protocols with 3 bits}.
\newblock \bibinfo{journal}{\emph{Distributed Comput.}} \bibinfo{volume}{32},
  \bibinfo{number}{3} (\bibinfo{year}{2019}), \bibinfo{pages}{173--191}.
\newblock
\urldef\tempurl%
\url{https://doi.org/10.1007/s00446-018-0330-x}
\showDOI{\tempurl}


\bibitem[\protect\citeauthoryear{Boczkowski, Natale, Feinerman, and
  Korman}{Boczkowski et~al\mbox{.}}{2018b}]%
        {DBLP:journals/ploscb/BoczkowskiNFK18}
\bibfield{author}{\bibinfo{person}{Lucas Boczkowski}, \bibinfo{person}{Emanuele
  Natale}, \bibinfo{person}{Ofer Feinerman}, {and} \bibinfo{person}{Amos
  Korman}.} \bibinfo{year}{2018}\natexlab{b}.
\newblock \showarticletitle{Limits on reliable information flows through
  stochastic populations}.
\newblock \bibinfo{journal}{\emph{PLoS Comput. Biol.}} \bibinfo{volume}{14},
  \bibinfo{number}{6} (\bibinfo{year}{2018}).
\newblock
\urldef\tempurl%
\url{https://doi.org/10.1371/journal.pcbi.1006195}
\showDOI{\tempurl}


\bibitem[\protect\citeauthoryear{Censor-Hillel, Haeupler, Kelner, and
  Maymounkov}{Censor-Hillel et~al\mbox{.}}{2012}]%
        {censor2012global}
\bibfield{author}{\bibinfo{person}{Keren Censor-Hillel},
  \bibinfo{person}{Bernhard Haeupler}, \bibinfo{person}{Jonathan Kelner}, {and}
  \bibinfo{person}{Petar Maymounkov}.} \bibinfo{year}{2012}\natexlab{}.
\newblock \showarticletitle{Global computation in a poorly connected world:
  fast rumor spreading with no dependence on conductance}. In
  \bibinfo{booktitle}{\emph{Proceedings of the forty-fourth annual ACM
  symposium on Theory of computing}}. \bibinfo{pages}{961--970}.
\newblock


\bibitem[\protect\citeauthoryear{Chen, Cummings, Doty, and Soloveichik}{Chen
  et~al\mbox{.}}{2014}]%
        {chen2014speed}
\bibfield{author}{\bibinfo{person}{Ho-Lin Chen}, \bibinfo{person}{Rachel
  Cummings}, \bibinfo{person}{David Doty}, {and} \bibinfo{person}{David
  Soloveichik}.} \bibinfo{year}{2014}\natexlab{}.
\newblock \showarticletitle{Speed faults in computation by chemical reaction
  networks}. In \bibinfo{booktitle}{\emph{International Symposium on
  Distributed Computing}}. Springer, \bibinfo{pages}{16--30}.
\newblock


\bibitem[\protect\citeauthoryear{Chierichetti, Giakkoupis, Lattanzi, and
  Panconesi}{Chierichetti et~al\mbox{.}}{2018}]%
        {chierichetti2018rumor}
\bibfield{author}{\bibinfo{person}{Flavio Chierichetti},
  \bibinfo{person}{George Giakkoupis}, \bibinfo{person}{Silvio Lattanzi}, {and}
  \bibinfo{person}{Alessandro Panconesi}.} \bibinfo{year}{2018}\natexlab{}.
\newblock \showarticletitle{Rumor spreading and conductance}.
\newblock \bibinfo{journal}{\emph{Journal of the ACM (JACM)}}
  \bibinfo{volume}{65}, \bibinfo{number}{4} (\bibinfo{year}{2018}),
  \bibinfo{pages}{1--21}.
\newblock


\bibitem[\protect\citeauthoryear{Cvikel, Berg, Levin, Hurme, Borissov, Boonman,
  Amichai, and Yovel}{Cvikel et~al\mbox{.}}{2015}]%
        {cvikel2015bats}
\bibfield{author}{\bibinfo{person}{Noam Cvikel}, \bibinfo{person}{Katya~Egert
  Berg}, \bibinfo{person}{Eran Levin}, \bibinfo{person}{Edward Hurme},
  \bibinfo{person}{Ivailo Borissov}, \bibinfo{person}{Arjan Boonman},
  \bibinfo{person}{Eran Amichai}, {and} \bibinfo{person}{Yossi Yovel}.}
  \bibinfo{year}{2015}\natexlab{}.
\newblock \showarticletitle{Bats aggregate to improve prey search but might be
  impaired when their density becomes too high}.
\newblock \bibinfo{journal}{\emph{Current Biology}} \bibinfo{volume}{25},
  \bibinfo{number}{2} (\bibinfo{year}{2015}), \bibinfo{pages}{206--211}.
\newblock


\bibitem[\protect\citeauthoryear{Danchin, Giraldeau, Valone, and
  Wagner}{Danchin et~al\mbox{.}}{2004}]%
        {danchin2004public}
\bibfield{author}{\bibinfo{person}{Etienne Danchin}, \bibinfo{person}{Luc-Alain
  Giraldeau}, \bibinfo{person}{Thomas~J Valone}, {and}
  \bibinfo{person}{Richard~H Wagner}.} \bibinfo{year}{2004}\natexlab{}.
\newblock \showarticletitle{Public information: from nosy neighbors to cultural
  evolution}.
\newblock \bibinfo{journal}{\emph{Science}} \bibinfo{volume}{305},
  \bibinfo{number}{5683} (\bibinfo{year}{2004}), \bibinfo{pages}{487--491}.
\newblock


\bibitem[\protect\citeauthoryear{Demers, Greene, Hauser, Irish, Larson,
  Shenker, Sturgis, Swinehart, and Terry}{Demers et~al\mbox{.}}{1987}]%
        {demers1987epidemic}
\bibfield{author}{\bibinfo{person}{Alan Demers}, \bibinfo{person}{Dan Greene},
  \bibinfo{person}{Carl Hauser}, \bibinfo{person}{Wes Irish},
  \bibinfo{person}{John Larson}, \bibinfo{person}{Scott Shenker},
  \bibinfo{person}{Howard Sturgis}, \bibinfo{person}{Dan Swinehart}, {and}
  \bibinfo{person}{Doug Terry}.} \bibinfo{year}{1987}\natexlab{}.
\newblock \showarticletitle{Epidemic algorithms for replicated database
  maintenance}. In \bibinfo{booktitle}{\emph{Proceedings of the sixth annual
  ACM Symposium on Principles of distributed computing}}.
  \bibinfo{pages}{1--12}.
\newblock


\bibitem[\protect\citeauthoryear{Dijkstra}{Dijkstra}{1974}]%
        {dijkstra}
\bibfield{author}{\bibinfo{person}{Edsger~W. Dijkstra}.}
  \bibinfo{year}{1974}\natexlab{}.
\newblock \showarticletitle{Self-stabilizing Systems in Spite of Distributed
  Control}.
\newblock \bibinfo{journal}{\emph{Commun. {ACM}}} \bibinfo{volume}{17},
  \bibinfo{number}{11} (\bibinfo{year}{1974}), \bibinfo{pages}{643--644}.
\newblock
\urldef\tempurl%
\url{https://doi.org/10.1145/361179.361202}
\showDOI{\tempurl}


\bibitem[\protect\citeauthoryear{Doerr, Goldberg, Minder, Sauerwald, and
  Scheideler}{Doerr et~al\mbox{.}}{2011}]%
        {doerr2011stabilizing}
\bibfield{author}{\bibinfo{person}{Benjamin Doerr}, \bibinfo{person}{Leslie~Ann
  Goldberg}, \bibinfo{person}{Lorenz Minder}, \bibinfo{person}{Thomas
  Sauerwald}, {and} \bibinfo{person}{Christian Scheideler}.}
  \bibinfo{year}{2011}\natexlab{}.
\newblock \showarticletitle{Stabilizing consensus with the power of two
  choices}. In \bibinfo{booktitle}{\emph{Proceedings of the twenty-third annual
  ACM symposium on Parallelism in algorithms and architectures}}.
  \bibinfo{pages}{149--158}.
\newblock


\bibitem[\protect\citeauthoryear{Dutta, Pandurangan, Rajaraman, Sun, and
  Viola}{Dutta et~al\mbox{.}}{2013}]%
        {dutta2013complexity}
\bibfield{author}{\bibinfo{person}{Chinmoy Dutta}, \bibinfo{person}{Gopal
  Pandurangan}, \bibinfo{person}{Rajmohan Rajaraman}, \bibinfo{person}{Zhifeng
  Sun}, {and} \bibinfo{person}{Emanuele Viola}.}
  \bibinfo{year}{2013}\natexlab{}.
\newblock \showarticletitle{On the complexity of information spreading in
  dynamic networks}. In \bibinfo{booktitle}{\emph{Proceedings of the
  Twenty-Fourth Annual ACM-SIAM Symposium on Discrete Algorithms}}. SIAM,
  \bibinfo{pages}{717--736}.
\newblock


\bibitem[\protect\citeauthoryear{Feinerman, Haeupler, and Korman}{Feinerman
  et~al\mbox{.}}{2017}]%
        {DBLP:journals/dc/FeinermanHK17}
\bibfield{author}{\bibinfo{person}{Ofer Feinerman}, \bibinfo{person}{Bernhard
  Haeupler}, {and} \bibinfo{person}{Amos Korman}.}
  \bibinfo{year}{2017}\natexlab{}.
\newblock \showarticletitle{Breathe before speaking: efficient information
  dissemination despite noisy, limited and anonymous communication}.
\newblock \bibinfo{journal}{\emph{Distributed Comput.}} \bibinfo{volume}{30},
  \bibinfo{number}{5} (\bibinfo{year}{2017}), \bibinfo{pages}{339--355}.
\newblock
\urldef\tempurl%
\url{https://doi.org/10.1007/s00446-015-0249-4}
\showDOI{\tempurl}


\bibitem[\protect\citeauthoryear{Feinerman and Korman}{Feinerman and
  Korman}{2017}]%
        {feinerman2017individual}
\bibfield{author}{\bibinfo{person}{Ofer Feinerman} {and} \bibinfo{person}{Amos
  Korman}.} \bibinfo{year}{2017}\natexlab{}.
\newblock \showarticletitle{Individual versus collective cognition in social
  insects}.
\newblock \bibinfo{journal}{\emph{Journal of Experimental Biology}}
  \bibinfo{volume}{220}, \bibinfo{number}{1} (\bibinfo{year}{2017}),
  \bibinfo{pages}{73--82}.
\newblock


\bibitem[\protect\citeauthoryear{Georgiou, Gilbert, Guerraoui, and
  Kowalski}{Georgiou et~al\mbox{.}}{2013}]%
        {DBLP:journals/jacm/GeorgiouGGK13}
\bibfield{author}{\bibinfo{person}{Chryssis Georgiou}, \bibinfo{person}{Seth
  Gilbert}, \bibinfo{person}{Rachid Guerraoui}, {and}
  \bibinfo{person}{Dariusz~R. Kowalski}.} \bibinfo{year}{2013}\natexlab{}.
\newblock \showarticletitle{Asynchronous gossip}.
\newblock \bibinfo{journal}{\emph{J. {ACM}}} \bibinfo{volume}{60},
  \bibinfo{number}{2} (\bibinfo{year}{2013}), \bibinfo{pages}{11:1--11:42}.
\newblock
\urldef\tempurl%
\url{https://doi.org/10.1145/2450142.2450147}
\showDOI{\tempurl}


\bibitem[\protect\citeauthoryear{Georgiou, Gilbert, and Kowalski}{Georgiou
  et~al\mbox{.}}{2011}]%
        {DBLP:journals/dc/GeorgiouGK11}
\bibfield{author}{\bibinfo{person}{Chryssis Georgiou}, \bibinfo{person}{Seth
  Gilbert}, {and} \bibinfo{person}{Dariusz~R. Kowalski}.}
  \bibinfo{year}{2011}\natexlab{}.
\newblock \showarticletitle{Meeting the deadline: on the complexity of
  fault-tolerant continuous gossip}.
\newblock \bibinfo{journal}{\emph{Distributed Comput.}} \bibinfo{volume}{24},
  \bibinfo{number}{5} (\bibinfo{year}{2011}), \bibinfo{pages}{223--244}.
\newblock
\urldef\tempurl%
\url{https://doi.org/10.1007/s00446-011-0144-6}
\showDOI{\tempurl}


\bibitem[\protect\citeauthoryear{Giakkoupis}{Giakkoupis}{2014}]%
        {giakkoupis2014tight}
\bibfield{author}{\bibinfo{person}{George Giakkoupis}.}
  \bibinfo{year}{2014}\natexlab{}.
\newblock \showarticletitle{Tight bounds for rumor spreading with vertex
  expansion}. In \bibinfo{booktitle}{\emph{Proceedings of the twenty-fifth
  annual ACM-SIAM symposium on Discrete algorithms}}. SIAM,
  \bibinfo{pages}{801--815}.
\newblock


\bibitem[\protect\citeauthoryear{Giakkoupis and Woelfel}{Giakkoupis and
  Woelfel}{2011}]%
        {giakkoupis2011randomness}
\bibfield{author}{\bibinfo{person}{George Giakkoupis} {and}
  \bibinfo{person}{Philipp Woelfel}.} \bibinfo{year}{2011}\natexlab{}.
\newblock \showarticletitle{On the randomness requirements of rumor spreading}.
  In \bibinfo{booktitle}{\emph{Proceedings of the Twenty-Second Annual ACM-SIAM
  Symposium on Discrete Algorithms}}. SIAM, \bibinfo{pages}{449--461}.
\newblock


\bibitem[\protect\citeauthoryear{Giraldeau and Caraco}{Giraldeau and
  Caraco}{2018}]%
        {giraldeau2018social}
\bibfield{author}{\bibinfo{person}{Luc-Alain Giraldeau} {and}
  \bibinfo{person}{Thomas Caraco}.} \bibinfo{year}{2018}\natexlab{}.
\newblock \bibinfo{booktitle}{\emph{Social foraging theory}}.
\newblock \bibinfo{publisher}{Princeton University Press}.
\newblock


\bibitem[\protect\citeauthoryear{Karp, Schindelhauer, Shenker, and
  Vocking}{Karp et~al\mbox{.}}{2000}]%
        {karp2000randomized}
\bibfield{author}{\bibinfo{person}{Richard Karp}, \bibinfo{person}{Christian
  Schindelhauer}, \bibinfo{person}{Scott Shenker}, {and}
  \bibinfo{person}{Berthold Vocking}.} \bibinfo{year}{2000}\natexlab{}.
\newblock \showarticletitle{Randomized rumor spreading}. In
  \bibinfo{booktitle}{\emph{Proceedings 41st Annual Symposium on Foundations of
  Computer Science}}. IEEE, \bibinfo{pages}{565--574}.
\newblock


\bibitem[\protect\citeauthoryear{Korman, Greenwald, and Feinerman}{Korman
  et~al\mbox{.}}{2014}]%
        {DBLP:journals/ploscb/KormanGF14}
\bibfield{author}{\bibinfo{person}{Amos Korman}, \bibinfo{person}{Efrat
  Greenwald}, {and} \bibinfo{person}{Ofer Feinerman}.}
  \bibinfo{year}{2014}\natexlab{}.
\newblock \showarticletitle{Confidence Sharing: An Economic Strategy for
  Efficient Information Flows in Animal Groups}.
\newblock \bibinfo{journal}{\emph{PLoS Comput. Biol.}} \bibinfo{volume}{10},
  \bibinfo{number}{10} (\bibinfo{year}{2014}).
\newblock
\urldef\tempurl%
\url{https://doi.org/10.1371/journal.pcbi.1003862}
\showDOI{\tempurl}


\bibitem[\protect\citeauthoryear{Liggett and Liggett}{Liggett and
  Liggett}{1985}]%
        {liggett1985interacting}
\bibfield{author}{\bibinfo{person}{Thomas~Milton Liggett} {and}
  \bibinfo{person}{Thomas~M Liggett}.} \bibinfo{year}{1985}\natexlab{}.
\newblock \bibinfo{booktitle}{\emph{Interacting particle systems}}.
  Vol.~\bibinfo{volume}{2}.
\newblock \bibinfo{publisher}{Springer}.
\newblock


\bibitem[\protect\citeauthoryear{Rajendran, Haluts, Gov, and
  Feinerman}{Rajendran et~al\mbox{.}}{2022}]%
        {rajendran2022ants}
\bibfield{author}{\bibinfo{person}{Harikrishnan Rajendran},
  \bibinfo{person}{Amir Haluts}, \bibinfo{person}{Nir~S Gov}, {and}
  \bibinfo{person}{Ofer Feinerman}.} \bibinfo{year}{2022}\natexlab{}.
\newblock \showarticletitle{Ants resort to majority concession to reach
  democratic consensus in the presence of a persistent minority}.
\newblock \bibinfo{journal}{\emph{Current Biology}} (\bibinfo{year}{2022}).
\newblock


\bibitem[\protect\citeauthoryear{Razin, Eckmann, and Feinerman}{Razin
  et~al\mbox{.}}{2013}]%
        {razin2013desert}
\bibfield{author}{\bibinfo{person}{Nitzan Razin}, \bibinfo{person}{Jean-Pierre
  Eckmann}, {and} \bibinfo{person}{Ofer Feinerman}.}
  \bibinfo{year}{2013}\natexlab{}.
\newblock \showarticletitle{Desert ants achieve reliable recruitment across
  noisy interactions}.
\newblock \bibinfo{journal}{\emph{Journal of the Royal Society Interface}}
  \bibinfo{volume}{10}, \bibinfo{number}{82} (\bibinfo{year}{2013}),
  \bibinfo{pages}{20130079}.
\newblock


\bibitem[\protect\citeauthoryear{Sumpter, Krause, James, Couzin, and
  Ward}{Sumpter et~al\mbox{.}}{2008}]%
        {sumpter2008consensus}
\bibfield{author}{\bibinfo{person}{David~JT Sumpter}, \bibinfo{person}{Jens
  Krause}, \bibinfo{person}{Richard James}, \bibinfo{person}{Iain~D Couzin},
  {and} \bibinfo{person}{Ashley~JW Ward}.} \bibinfo{year}{2008}\natexlab{}.
\newblock \showarticletitle{Consensus decision making by fish}.
\newblock \bibinfo{journal}{\emph{Current Biology}} \bibinfo{volume}{18},
  \bibinfo{number}{22} (\bibinfo{year}{2008}), \bibinfo{pages}{1773--1777}.
\newblock


\bibitem[\protect\citeauthoryear{Wilkinson}{Wilkinson}{1992}]%
        {wilkinson1992information}
\bibfield{author}{\bibinfo{person}{Gerald~S Wilkinson}.}
  \bibinfo{year}{1992}\natexlab{}.
\newblock \showarticletitle{Information transfer at evening bat colonies}.
\newblock \bibinfo{journal}{\emph{Animal Behaviour}}  \bibinfo{volume}{44}
  (\bibinfo{year}{1992}), \bibinfo{pages}{501--518}.
\newblock


\bibitem[\protect\citeauthoryear{Yick, Mukherjee, and Ghosal}{Yick
  et~al\mbox{.}}{2008}]%
        {yick2008wireless}
\bibfield{author}{\bibinfo{person}{Jennifer Yick}, \bibinfo{person}{Biswanath
  Mukherjee}, {and} \bibinfo{person}{Dipak Ghosal}.}
  \bibinfo{year}{2008}\natexlab{}.
\newblock \showarticletitle{Wireless sensor network survey}.
\newblock \bibinfo{journal}{\emph{Computer networks}} \bibinfo{volume}{52},
  \bibinfo{number}{12} (\bibinfo{year}{2008}), \bibinfo{pages}{2292--2330}.
\newblock


\end{thebibliography}
\clearpage


\centerline{\huge{Appendix}}
\appendix
\section{Probabilistic tools}
\subsection{Some well-known theorems}
\begin{theorem} \label{thm:mult_chernoff_bound} [Multiplicative Chernoff's Bound]
    Let~$X_1,\ldots,X_n$ be independent binary random variables, let~$X = \sum_{i=1}^n X_i$ and~$\mu = \bbE(X)$. Then it holds for all~$\delta > 0$ that
    \begin{equation*}
        \bbP \pa{ X \geq (1+\delta) \mu } \leq \exp \pa{ - \min\{ \delta, \delta^2 \} \cdot \frac{\mu}{3} },
    \end{equation*}
    and for all~$0 < \epsilon < 1$,
    \begin{equation*}
         \bbP \pa{ X \leq (1-\epsilon) \mu } \leq \exp \pa{-\epsilon^2 \cdot \frac{\mu}{2}}.
    \end{equation*}
\end{theorem}


\begin{theorem} \label{thm:hoeffding_bound} [H{\oe}ffding's bound]
    Let~$X_1,\ldots,X_n$ be independent random variables such that for every~$1 \leq i \leq n$, $a_i \leq X_i \leq b_i$ almost surely. Let~$X = \sum_{i=1}^n X_i$ and~$\mu = \bbE(X)$. Then it holds for all~$\delta > 0$ that
    \begin{equation*}
        \bbP \pa{ X-\mu \geq \delta } \leq \exp \pa{ - \frac{2 \delta^2}{\sum_{i=1}^n (b_i-a_i)^2 } }.
    \end{equation*}
\end{theorem}


\begin{theorem} \label{thm:central_limit} [Central Limit]
    Let~$X_1,\ldots,X_n$ be i.i.d. random variables with~$\bbE(X_1)=\mu$ and~$\var(X_1) = \sigma^2 < +\infty$. Then as~$n$ tends to infinity, the random variables~$\sqrt{n} \pa{ \frac{1}{n} \sum_{i=1}^n X_i - \mu}$ converges in distribution to~$\calN(0,\sigma^2)$.
\end{theorem}

Let~$\Phi$ be the cumulative distribution function (c.d.f.) of the standard normal distribution:
\begin{equation*}
    \Phi(x) =  \frac{1}{\sqrt{2\pi}} \int_{-\infty}^x e^{-t^2/2} dt.
\end{equation*}

\begin{theorem} \label{thm:berry_esseen} [Berry-Esseen]
    Let~$X_1,\ldots,X_n$ be i.i.d. random variables, with~$\bbE(X_1)=0$, $\var(X_1) = \bbE(X_1^2) > 0$, and~$\bbE(|X_1|^3) = \rho < +\infty$. Let~$X = \sum_{i=1}^n X_i$ and~$F$ be the c.d.f. of~$X/(\sigma \sqrt{n})$.
    Then it holds that
    \begin{equation*}
        |F(x) - \Phi(x)| \leq \frac{C \rho}{\sigma^3 \sqrt{n}},
    \end{equation*}
    where, e.g.,~$C = 0.4748$.
\end{theorem}

\subsection{Competition between coins}

Consider two coins such that one coin has a greater probability of yielding ``heads'', and  toss them $k$ times each.

\subsubsection{Lower bounds on the probability that the best coin wins}

In Lemmas \ref{lem:topdog} and \ref{lem:handcrafted} we 
aim to lower bound the probability that the more likely coin yields more ``heads'', or in other words, we lower bound the probability that the favorite coin wins. Lemma \ref{lem:topdog} is particularly effective when the difference between $p$ and $q$ is sufficiently large. Its proof is based on a simple application of H{\oe}ffding's inequality.
\begin{lemma} \label{lem:topdog}
    For every~$p,q \in [0,1]$ s.t.~$p < q$ and every integer~$k$, we have
    \begin{equation*} 
        \prob{k}{p}{<}{k}{q} \geq  1- \exp \pa{-\frac{1}{2} k (q - p)^2}.
    \end{equation*}
\end{lemma}
\begin{proof}
    Let~$Y_i$, $i \in \{ 1,\ldots,k \}$ be i.i.d. random variables with
    \begin{equation*}
        Y_i =
        \begin{cases}
            1 & \text{w.p. } p(1-q), \\
            0 & \text{w.p. } p q + (1-p)(1-q), \\
            -1 & \text{w.p. } (1-p)q.
        \end{cases}
    \end{equation*}
    Then
    \begin{align*}
        \prob{k}{p}{\geq}{k}{q} = \bbP \pa{ \sum_{i=1}^k Y_i \geq 0 } = \bbP \pa{ \frac{1}{k} \sum_{i=1}^k \pa{Y_i - (p - q)} \geq (q - p) }.    \end{align*}
    Since each~$Y_i$ is bounded, and~$\bbE(Y_i) = (p - q)$, we can apply H{\oe}ffding's inequality (Theorem~\ref{thm:hoeffding_bound}) to get
    \begin{equation*}
        \prob{k}{p}{\geq}{k}{q} \leq \exp \pa{- \frac{2k^2 (q - p)^2}{4k}} =  \exp \pa{-\frac{1}{2} k (q - p)^2 }.
    \end{equation*}
\end{proof}
Lemma \ref{lem:topdog} is not particularly effective when $p$ and $q$ are close to each other. For such cases, we shall use the following lemma.
\begin{lemma} \label{lem:handcrafted}
    Let~$\lambda>0$. There exist~$\epsilon = \epsilon(\lambda)$ and $K = K(\lambda)$, s.t. for every~$p,q \in [1/2-\epsilon,1/2+\epsilon]$ with~$p<q$, and every~$k > K$,
    \begin{equation*}
        \prob{k}{p}{<}{k}{q} > \frac{1}{2} + \lambda \cdot (q-p) - \frac{1}{2}\bbP\pa{ B_k(p) = B_k(q)}.
    \end{equation*}
\end{lemma}
\begin{proof}
    For every~$p,q \in [0,1]$, we have
    \begin{align*}
        \prob{k}{q}{<}{k}{p} &= \sum_{d=1}^k \bbP \pa{B_k(q) = B_k(p) - d} \\
        &= \sum_{d=1}^k  \bbP \pa{|B_k(q) - B_k(p)| = d} \cdot \frac{\bbP \pa{B_k(q) = B_k(p) - d}}{\bbP \pa{|B_k(q) - B_k(p)| = d}} \\
        &= \sum_{d=1}^k  \bbP \pa{|B_k(q) - B_k(p)| = d} \cdot \frac{\bbP \pa{B_k(q) = B_k(p) - d}}{\bbP \pa{B_k(q) = B_k(p) - d}+\bbP \pa{B_k(p) = B_k(q) - d}},
    \end{align*}
    so
    \begin{multline} \label{eq:conditioning}
        \prob{k}{p}{<}{k}{q} - \prob{k}{q}{<}{k}{p} \\= \sum_{d=1}^k  \bbP \pa{|B_k(q) - B_k(p)| = d} \cdot \frac{\bbP \pa{B_k(p) = B_k(q) - d} - \bbP \pa{B_k(q) = B_k(p) - d}}{\bbP \pa{B_k(p) = B_k(q) - d}+\bbP \pa{B_k(q) = B_k(p) - d}}.
    \end{multline}
    Let us compute~$\bbP \pa{B_k(q) = B_k(p) - d}$:
    \begin{align*}
        \bbP \pa{B_k(q) = B_k(p) - d} &= \sum_{i=0}^{k-d} \bbP \pa{B_k(q) = i} \cdot \bbP \pa{B_k(p) = i+d} \\
        &= \sum_{i=0}^{k-d} \binom{k}{i} \binom{k}{i+d} q^i (1-q)^{k-i} p^{i+d} (1-p)^{k-i-d} \\
        &= \pa{ p(1-q) }^d \sum_{i=0}^{k-d} \binom{k}{i} \binom{k}{i+d} (qp)^i ((1-q)(1-p))^{k-i-d} \\
        &:= \pa{ p(1-q) }^d A_{k,d,p,q},
    \end{align*}
    where
    \begin{equation*}
        A_{k,d,p,q} := \sum_{i=0}^{k-d} \binom{k}{i} \binom{k}{i+d} (qp)^i  ((1-q)(1-p))^{k-i-d}.
    \end{equation*}
    Since~$A_{k,d,p,q}$ is symmetric w.r.t.~$p,q$, i.e., $A_{k,d,p,q} = A_{k,d,q,p}$, we can simplify Eq.~\eqref{eq:conditioning} as
    \begin{equation} \label{eq:conditioning_simp}
        \bbP \pa{B_k(p) < B_k(q)} - \bbP \pa{B_k(q) < B_k(p)} = \sum_{d=1}^k  \bbP \pa{|B_k(q) - B_k(p)| = d} \cdot \frac{\pa{ q(1-p) }^d - \pa{ p(1-q) }^d}{\pa{ q(1-p) }^d + \pa{ p(1-q) }^d}.
    \end{equation}
    Intuitively, the quantity
    \begin{equation*}
        \frac{\pa{ q(1-p) }^d - \pa{ p(1-q) }^d}{\pa{ q(1-p) }^d + \pa{ p(1-q) }^d}
    \end{equation*}
    can be seen as the ``advantage'' given by playing with the better coin ($q$) in a $k$-coin-tossing contest, knowing that one coin hit ``head'' $d$ times more than the other. Before we continue, we need the following simple claim. 
    \begin{claim} \label{claim:incr_sequ}
        For every~$a,b \in [0,1]$ with~$a>b$, the sequence
        \begin{equation*}
            \pa{\frac{a^n - b^n}{a^n + b^n} }, n \in \bbN
        \end{equation*}
        is increasing in~$n$.
    \end{claim}
    \begin{proof}
        Rewrite
        \begin{equation*}
            \frac{a^n - b^n}{a^n + b^n} = \frac{2a^n}{a^n+b^n}-1 = 2\cdot\frac{1}{1+(b/a)^n}-1,
        \end{equation*}
        and notice that, since~$a>b$, $\pa{(b/a)^n}$,~$n \in \bbN$ is a decreasing sequence.
    \end{proof}
    \begin{claim} \label{claim:lb_sequ}
        Let~$0 < \gamma < 1$ and~$d\in \bbN$. There exists~$\epsilon = \epsilon(\gamma,d)$, such that for every~$p,q \in [1/2-\epsilon,1/2+\epsilon]$ with~$p<q$,
        \begin{equation*}
            \frac{\pa{ q(1-p) }^d - \pa{ p(1-q) }^d}{\pa{ q(1-p) }^d + \pa{ p(1-q) }^d} > (q-p)\cdot 2d\gamma.
        \end{equation*}
    \end{claim}
    \begin{proof}
        First, by a telescopic argument:
        \begin{align*}
            \pa{ q(1-p) }^d - \pa{p(1-q) }^d &= \pa{q(1-p) - p(1-q)} \sum_{i=0}^{d-1} \pa{ q(1-p) }^{d-1-i} \cdot \pa{ p(1-q) }^i \\
            &= \pa{q-p} \sum_{i=0}^{d-1} \pa{ q(1-p) }^{d-1-i} \cdot \pa{ p(1-q) }^i.
        \end{align*}
        Note that
        \begin{equation*}
            \lim_{p,q \rightarrow 1/2} \sum_{i=0}^{d-1} \pa{ q(1-p) }^{d-1-i} \cdot \pa{ p(1-q) }^i = \sum_{i=0}^{d-1} \pa{ \frac{1}{4} }^{d-1-i} \cdot \pa{ \frac{1}{4} }^i = \sum_{i=0}^{d-1} \pa{\frac{1}{2}}^{2d-2} = d\cdot \pa{\frac{1}{2}}^{2d-2},
        \end{equation*}
        and that
        \begin{equation*}
            \lim_{p,q \rightarrow 1/2} \pa{ q(1-p) }^d + \pa{ p(1-q) }^d = \pa{ \frac{1}{4} }^d + \pa{ \frac{1}{4} }^d = \pa{\frac{1}{2}}^{2d-1}.
        \end{equation*}
        Hence, since~$\gamma<1$, and provided that~$p,q$ are close enough to~$1/2$, we obtain 
        \begin{equation*}
            \frac{\pa{ q(1-p) }^d - \pa{ p(1-q) }^d}{\pa{ q(1-p) }^d + \pa{ p(1-q) }^d} > (q-p)\cdot 2d\gamma,
        \end{equation*}
        which completes the proof of Claim \ref{claim:lb_sequ}.
    \end{proof}
    Next, let~$\lambda>0$ as in the Lemma's statement, and let~$\lambda' = \lambda+1$.
    Denote~$D = \lceil \lambda' \rceil+1 > \lambda'$ and~$\gamma = \lambda'/D < 1$. By Claim~\ref{claim:lb_sequ}, there exists~$\epsilon = \epsilon(\gamma,D) = \epsilon(\lambda)$, s.t. for~$p,q \in [1/2-\epsilon,1/2+\epsilon]$,
    \begin{equation} \label{eq:apply_previous_claim}
        \frac{\pa{ q(1-p) }^D - \pa{ p(1-q) }^D}{\pa{ q(1-p) }^D + \pa{ p(1-q) }^D} > (q-p)\cdot 2\lambda'.
    \end{equation}
    Now we derive a lower bound on Eq.~\eqref{eq:conditioning_simp}:
    \begin{align*}
        \bbP \pa{B_k(p) < B_k(q)} - \bbP \pa{B_k(q) < B_k(p)} &= \sum_{d=1}^k  \bbP \pa{|B_k(q) - B_k(p)| = d} \cdot \frac{\pa{ q(1-p) }^d - \pa{ p(1-q) }^d}{\pa{ q(1-p) }^d + \pa{ p(1-q) }^d} & \text{(Eq.~\eqref{eq:conditioning_simp})}\\
        &\geq \sum_{d=D}^k  \bbP \pa{|B_k(q) - B_k(p)| = d} \cdot \frac{\pa{ q(1-p) }^d - \pa{ p(1-q) }^d}{\pa{ q(1-p) }^d + \pa{ p(1-q) }^d} & \\
        &\geq \sum_{d=D}^k  \bbP \pa{|B_k(q) - B_k(p)| = d} \cdot \frac{\pa{ q(1-p) }^D - \pa{ p(1-q) }^D}{\pa{ q(1-p) }^D + \pa{ p(1-q) }^D} & \text{(by Claim~\ref{claim:incr_sequ})}\\
        &> (q-p)\cdot 2\lambda' \sum_{d=D}^k  \bbP \pa{|B_k(q) - B_k(p)| = d} & \text{(by Eq.~\eqref{eq:apply_previous_claim})}\\
        &= (q-p)\cdot 2\lambda' \cdot \pa{ 1 - \bbP \pa{|B_k(q) - B_k(p)| < D} }. &
    \end{align*}
    Since~$\lambda'>\lambda$, and since~$\bbP \pa{|B_k(q) - B_k(p)| < D}$ tends to~$0$ as~$k$ tends to~$+\infty$, there exists~$K = K(\lambda)$ s.t. for all~$k>K$,
    \begin{equation} \label{eq:final_eq}
        \bbP \pa{B_k(p) < B_k(q)} - \bbP \pa{B_k(q) < B_k(p)} > (q-p)\cdot 2\lambda.
    \end{equation}
    Eventually, we write
    \begin{align*}
        \bbP \pa{B_k(p) < B_k(q)} &= 1 - \bbP \pa{B_k(q) < B_k(p)} - \bbP\pa{ B_k(p) = B_k(q) } &\\
        &> 1 - \bbP \pa{B_k(p) < B_k(q)} + 2\lambda(q-p) - \bbP\pa{ B_k(p) = B_k(q)}. & \text{(by Eq.~\eqref{eq:final_eq})}
    \end{align*}
    Hence,
    \begin{equation*}
        \bbP \pa{B_k(p) < B_k(q)} > \frac{1}{2} + \lambda(q-p) - \frac{1}{2}\bbP\pa{ B_k(p) = B_k(q)},
    \end{equation*}
    which concludes the proof of Lemma \ref{lem:handcrafted}.
\end{proof}

\subsubsection{Lower bounds on the probability that the worse coin wins} \label{app:underdog}

We now deal with the opposite problem, that is, to lower bound the probability that the underdog coin wins. Formally,
\begin{lemma} \label{lem:underdog}
    For every~$p,q \in [0,1]$ s.t.~$p < q$ and every integer~$k$, we have
    \begin{equation*}
        \bbP \pa{ B_k(p) > B_k(q) } \geq 1-\Phi\pa{\frac{\sqrt{k} (q-p)}{\sigma}} - \frac{C}{\sigma \sqrt{k}},
    \end{equation*}
    where~$C = 0.4748$ and~$\sigma = \sqrt{p(1-p) + q(1-q)}$.
\end{lemma}
\begin{proof}
    Let~$Y_i$, $i \in \{ 1,\ldots,k \}$ be i.i.d. random variables with
    \begin{equation*}
        Y_i =
        \begin{cases}
            1 & \text{w.p. } p(1-q), \\
            0 & \text{w.p. } p q + (1-p)(1-q), \\
            -1 & \text{w.p. } (1-p)q.
        \end{cases}
    \end{equation*}
    Let~$\mu = \bbE(Y_1) = (p-q)$, $\sigma = \sqrt{\var(Y_1)} = \sqrt{p(1-p) + q(1-q)}$, and~$ \rho = \bbE(|Y_1-\mu|^3)$.
    We have
    \begin{align*}
        \bbP \pa{ B_k(p) > B_k(q) } = \bbP \pa{ \sum_{i=1}^k Y_i > 0 } &= \bbP \pa{ \frac{1}{\sqrt{k}} \sum_{i=1}^k \pa{Y_i - (p-q)} > \sqrt{k} (q-p)} \\
        &= \bbP \pa{ \frac{1}{\sigma \sqrt{k}} \sum_{i=1}^k \pa{Y_i - (p-q)} > \frac{\sqrt{k} (q-p)}{\sigma}} \\
        &= \bbP \pa{ Z > \frac{\sqrt{k} (q-p)}{\sigma}},
    \end{align*}
    where
    \begin{equation*}
        Z = \frac{1}{\sigma \sqrt{k}} \sum_{i=1}^k \pa{Y_i - (p-q)}.
    \end{equation*}
    By the Berry-Esseen theorem (Theorem~\ref{thm:berry_esseen}),
    \begin{equation*}
        \left| \bbP \pa{ Z < \frac{\sqrt{k} (q-p)}{\sigma}} - \Phi \pa{\frac{\sqrt{k} (q-p)}{\sigma}} \right| < \frac{C\rho}{\sigma^3 \sqrt{k}},
    \end{equation*}
    implying that
    \begin{equation*}
        \left| \left( 1 - \Phi \pa{\frac{\sqrt{k} (q-p)}{\sigma}} \right)-\bbP \pa{ Z > \frac{\sqrt{k} (q-p)}{\sigma}} \right| < \frac{C\rho}{\sigma^3 \sqrt{k}},
    \end{equation*}
    and so
    \begin{equation*}
        \bbP \pa{ Z > \frac{\sqrt{k} (q-p)}{\sigma}} > 1-\Phi\pa{\frac{\sqrt{k} (q-p)}{\sigma}} - \frac{C\rho}{\sigma^3 \sqrt{k}},
    \end{equation*}
    where, e.g., $C = 0.4748$. 
    Since~$\rho < \sigma^2$, we end up with
    \begin{equation*}
        \bbP \pa{ B_k(p) > B_k(q) } \geq 1-\Phi\pa{\frac{\sqrt{k} (q-p)}{\sigma}} - \frac{C}{\sigma \sqrt{k}},
    \end{equation*}
    which concludes the proof.
\end{proof}
Just as Lemma~\ref{lem:handcrafted} was a version of Lemma~\ref{lem:topdog} optimized for cases where~$p$ and~$q$ are close to each other, Lemma~\ref{lem:handcrafted_converse} (stated in Section~\ref{sec:lem:B1}) complements Lemma~\ref{lem:underdog} in such situations.
\begin{proof} [Proof of Lemma~\ref{lem:handcrafted_converse}]
    Recall that (see Eq.~\eqref{eq:conditioning_simp} in the proof of Lemma~\ref{lem:handcrafted}):
    \begin{equation} \label{eq:conditioning_simp_recall}
        \bbP \pa{B_k(p) < B_k(q)} - \bbP \pa{B_k(q) < B_k(p)} = \sum_{d=1}^k  \bbP \pa{|B_k(q) - B_k(p)| = d} \cdot \frac{\pa{ q(1-p) }^d - \pa{ p(1-q) }^d}{\pa{ q(1-p) }^d + \pa{ p(1-q) }^d}.
    \end{equation}
    The following claim is analogous to Claim~\ref{claim:lb_sequ}, but this time we are looking for an upper bound (instead of a lower bound) on the same quantity.
    \begin{claim} \label{claim:ub_sequ}
        There exists a constant~$\alpha>1$, s.t. for every integer~$k$, every~$p,q \in [1/3,2/3]$ with~$p<q$, and all~$d\in \bbN$,
        \begin{equation*}
            \frac{\pa{ q(1-p) }^d - \pa{ p(1-q) }^d}{\pa{ q(1-p) }^d + \pa{ p(1-q) }^d} < \alpha d \cdot (q-p).
        \end{equation*}
    \end{claim}
    \begin{proof}
        First, we use a well-known identity:
        \begin{align*}
            \pa{ q(1-p) }^d - \pa{ p(1-q) }^d &= \pa{q(1-p) - p(1-q)} \sum_{i=0}^{d-1} \pa{ q(1-p) }^{d-1-i} \cdot \pa{ p(1-q) }^i \\
            &= \pa{q-p} \sum_{i=0}^{d-1} \pa{ q(1-p) }^{d-1-i} \cdot \pa{ p(1-q) }^i \\
            &\leq d\cdot (q-p) \pa{ q(1-p) }^{d-1} \\
            &\leq \alpha d\cdot (q-p) \pa{ q(1-p) }^d,
        \end{align*}
        where~$\alpha$ is any upper bound on~$1/(q(1-p))$, e.g., $\alpha = 9$. Hence,
        \begin{equation*}
            \frac{\pa{ q(1-p) }^d - \pa{ p(1-q) }^d}{\pa{ q(1-p) }^d + \pa{ p(1-q) }^d} \leq \alpha d\cdot (q-p) \cdot \frac{\pa{ q(1-p) }^d}{\pa{ q(1-p) }^d + \pa{ p(1-q) }^d} \leq \alpha d \cdot (q-p),
        \end{equation*}
        which concludes the proof of Claim~\ref{claim:ub_sequ}.
    \end{proof}
    Using Claim~\ref{claim:ub_sequ} on Eq.~\eqref{eq:conditioning_simp_recall}, we obtain
    \begin{equation} \label{eq:conditioning_ub}
        \bbP \pa{B_k(p) < B_k(q)} - \bbP \pa{B_k(q) < B_k(p)} \leq \alpha \cdot(q-p) \sum_{d=1}^k  d \cdot \bbP \pa{|B_k(q) - B_k(p)| = d} .
    \end{equation}
    \begin{claim} \label{lem:difference_binomial}
        For every~$p,q \in [1/3,2/3]$ with~$p<q$, and every integer~$k$,
        \begin{equation*}
            \bbE \pa{ |B_k(p) - B_k(q)| } \leq \sqrt{2kq(1-q)} + k\cdot(q-p).
        \end{equation*}
    \end{claim}
    \begin{proof}
        For~$i\in \{ 1,\ldots,k \}$, let~$X_i^{(1)},X_i^{(2)} \sim \calB(q)$ and~$Y_i \sim \calB(1-p/q)$ be independent random variables. Let
        \begin{equation*}
            X^{(1)} = \sum_{i=1}^k X_i^{(1)}, \qquad
            X^{(2)} = \sum_{i=1}^k X_i^{(2)}, \qquad
            Z = \sum_{i=1}^k X_i^{(2)} \cdot Y_i, \qquad
            \text{and}~\tilde{X}^{(2)} = \sum_{i=1}^k X_i^{(2)} \cdot (1-Y_i) = X^{(2)}-Z.
        \end{equation*}
        Clearly, $X^{(1)} \sim \calB_k(q)$ and~$X^{(2)} \sim \calB_k(q)$. Since for every~$i$,
        \begin{equation*}
            X_i^{(2)} \cdot (1-Y_i) = \begin{cases}
                1 & \text{if } X_i^{(2)} = 1 \text{ and } Y_i = 0, \\
                0 & \text{otherwise,}
            \end{cases} 
        \end{equation*}
         we obtain that~$\tilde{X}^{(2)} \sim \calB_k(q\cdot(1- (1-p/q))) = \calB_k(p)$.
        Similarly, for every~$i$,
        \begin{equation*}
            X_i^{(2)} \cdot Y_i = \begin{cases}
                1 & \text{if } X_i^{(2)} = 1 \text{ and } Y_i = 1, \\
                0 & \text{otherwise,}
            \end{cases} 
        \end{equation*}
        hence, we obtain that~$Z \sim \calB_k(q\cdot(1-p/q)) = \calB_k(q-p)$.
        We notice that~$(X^{(1)},X^{(2)})$ are independent, as well as~$(X^{(1)},\tilde{X}^{(2)})$. Hence
        \begin{align*}
            \bbE \pa{ |B_k(q) - B_k(p)| } &= \bbE \pa{ |X^{(1)} - \tilde{X}^{(2)}| } & \\
            &= \bbE \pa{ |X^{(1)} - X^{(2)} + Z| } &\\
            &\leq \bbE \pa{ |X^{(1)} - X^{(2)}| + Z } &\\
            &= \bbE \pa{ |X^{(1)} - X^{(2)}| } + \bbE(Z).&
        \end{align*}
        We have~$\bbE(Z) = k(q-p)$, and
        \begin{align*}
            \bbE \pa{ |X^{(1)} - X^{(2)}| } &= \bbE \pa{\sqrt{ \pa{X^{(1)} - X^{(2)}}^2 }} & \\
            &\leq \sqrt{ \bbE \pa{ \pa{X^{(1)} - X^{(2)}}^2 }} & \text{(by Jensen inequality, and the fact that $g(x)=x^2$ is convex)} \\
            &= \sqrt{\var\pa{X^{(1)} - X^{(2)}}} & \text{(since~$\bbE\pa{X^{(1)} - X^{(2)}}=0$)} \\
            &= \sqrt{2kq(1-q)}, & \text{($X^{(1)},X^{(2)}\sim\calB_k(q)$ and are independent)}.
        \end{align*}
        which concludes the proof of Claim~\ref{lem:difference_binomial}.
    \end{proof}
    We note that
    \begin{align*}
        \sum_{d=1}^k  d \cdot \bbP \pa{|B_k(q) - B_k(p)| = d} &= \bbE\pa{ |B_k(q) - B_k(p)| } &\\
        &\leq \sqrt{2kq(1-q)} + k\cdot(q-p) & (\text{by Claim~\ref{lem:difference_binomial})} \\
        &\leq \sqrt{2kq(1-q)} + \sqrt{k} & \text{(since~$q-p \leq 1/ \sqrt{k}$)} \\
        &\leq 2\sqrt{k}. &
    \end{align*}
    Eventually, Eq.~\eqref{eq:conditioning_ub} becomes
    \begin{equation} \label{eq:conditioning_ub_final}
        \bbP \pa{B_k(p) < B_k(q)} - \bbP \pa{B_k(q) < B_k(p)} \leq 2\alpha \cdot(q-p) \sqrt{k}.
    \end{equation}
    To conclude, we write
    \begin{align*}
        \bbP \pa{B_k(p) < B_k(q)} &= 1 - \bbP \pa{B_k(q) < B_k(p)} - \bbP\pa{ B_k(p) = B_k(q) } &\\
        &< 1 - \bbP \pa{B_k(p) < B_k(q)} + 2\alpha \cdot(q-p) \sqrt{k} - \bbP\pa{ B_k(p) = B_k(q)}. & \text{(by Eq.~\eqref{eq:conditioning_ub_final})}
    \end{align*}
    Hence,
    \begin{equation*}
        \bbP \pa{B_k(p) < B_k(q)} < \frac{1}{2} + \alpha \cdot(q-p) \sqrt{k} - \frac{1}{2}\bbP\pa{ B_k(p) = B_k(q)},
    \end{equation*}
    which concludes the proof of Lemma \ref{lem:handcrafted_converse}.
\end{proof}

\subsection{Proof of Claim~\ref{claim:last_minute_obstacle}} \label{app:last_minute_obstacle}

    The goal of this section is to prove Claim~\ref{claim:last_minute_obstacle}, whose statement can be found in Section~\ref{sec:lem:B1}. For this purpose, the key observation is that the derivative, w.r.t.~$x$, of~$\prob{k}{x}{>}{k}{p}$ in the neighborhood of~$p$ is relatively high. The following claim formalizes this idea.
    
    \begin{claim} \label{claim:derivative}
        There exists a constant~$\beta' > 0$ such that for every~$k$ large enough,
        and every~$p,x \in [1/3,2/3]$ satisfying~$p \leq x \leq p + 1/ \sqrt{k}$,
        \begin{equation*}
            \frac{d}{dx} \prob{k}{x}{>}{k}{p} \geq \beta' \cdot \sqrt{k}.
        \end{equation*}
    \end{claim}
    \begin{proof} 
        Let~$h > 0$.
        We will proceed by using a coupling argument. Let~$X_i$, $i\in\{ 1,\ldots,k \}$, be i.i.d. random variables uniformly distributed over the interval~$[0,1]$.
        Let~$Y_1 = |\{ i \text{ s.t. } X_i \leq x \}|$ and~$Y_2 = |\{ i \text{ s.t. } X_i \leq x+h \}|$. By construction, $Y_1 \sim \calB_k(x)$ and~$Y_2 \sim \calB_k(x+h)$.  Next, let~$H = |\{ i \text{ s.t. } x < X_i \leq x+h \}|$. By construction, $Y_2 = Y_1 + H \geq Y_1$. Let~$Z \sim \calB_k(p)$ be a binomially distributed random variable, independent from~$Y_1$ and~$Y_2$. Now, we have:
        \begin{align*}
            \prob{k}{x+h}{>}{k}{p} - \prob{k}{x}{>}{k}{p}
            &= \bbP \pa{ Y_2 > Z } - \bbP \pa{ Y_1 > Z } & \text{(by definition of~$Y_1, Y_2$ and~$Z$)} \\
            &= \bbP \pa{ Y_1 \leq Z \medcap Y_2 > Z } & \text{(because~$Y_1 > Z \Rightarrow Y_2 > Z$)} \\
            &= \sum_{j=0}^k \bbP\pa{ Z = j } \cdot \bbP\pa{ Y_1 \leq j \medcap Y_2 > j }. &
        \end{align*}
        Let~$J = \{ j \in \bbN \text{ s.t. } kp \leq j \leq kp + \sqrt{k} \}$.
        We can rewrite the last equation as
        \begin{equation} \label{eq:derivative_aux}
            \prob{k}{x+h}{>}{k}{p} - \prob{k}{x}{>}{k}{p} \geq \sum_{j \in J} \bbP\pa{ Z = j } \cdot \bbP\pa{ Y_1 \leq j \medcap Y_2 > j }.
        \end{equation}
        The following result is a well-known fact.
        \begin{observation} \label{claim:proba_equal}
            There exists a constant~$\beta > 0$ such that for every~$k$ large enough,
             every~$p \in [1/3,2/3]$, and every~$i$ satisfying~$|i-kp| \leq \sqrt{k}$, we have 
           $\bbP \pa{B_k(p) = i} \geq \frac{\beta}{\sqrt{k}}$.
        \end{observation}
        \begin{proof}
            By the De Moivre-Laplace theorem, for any~$i \in \{0,\ldots,k\}$,
            \begin{equation} \label{eq:demoivre-laplace}
                \bbP \pa{ B_k(p)=i } = \binom{k}{i} p^i (1-p)^{k-i} \approx \frac{1}{\sqrt{2kp(1-p)}} \exp\pa{ - \frac{(i-kp)^2}{2kp(1-p)}},
            \end{equation}
            where we used~$\approx$ in the sense that the ratio between the left-hand side and the right-hand side tends to~$1$ as~$k$ tends to infinity. Since~$|i-kp| \leq \sqrt{k}$,
            \begin{equation*}
                \frac{1}{\sqrt{2kp(1-p)}} \exp\pa{ - \frac{(i-kp)^2}{2kp(1-p)}} \geq \frac{1}{\sqrt{2kp(1-p)}} \exp\pa{ - \frac{1}{2p(1-p)}} := \frac{f(p)}{\sqrt{k}}.
            \end{equation*}
            By Eq.~\eqref{eq:demoivre-laplace}, we can conclude the proof of Observation~\ref{claim:proba_equal} for~$k$ large enough by taking, e.g.,
            \begin{equation*}
                \beta = \frac{1}{2} \cdot \min_{p \in [1/3,2/3]} f(p).
            \end{equation*}
        \end{proof}
        For~$j \in J$, by Observation~\ref{claim:proba_equal}, $\bbP\pa{ Z = j } \geq \beta / \sqrt{k}$, for some constant $\beta>0$. Moreover,
        \begin{align*}
            \bbP\pa{ Y_1 \leq j \medcap Y_2 > j } &\geq \bbP\pa{ Y_1 = j \medcap Y_2 > j } &\\
            &= \bbP\pa{ Y_1 = j \medcap H \geq 1 } & \text{(because~$Y_2=Y_1+H$)} \\
            &= \bbP\pa{ Y_1 = j} \cdot \bbP\pa{ H \geq 1 \mid Y_1 = j}. &
        \end{align*}
        By the assumption in the lemma, $p \leq x \leq p + 1/ \sqrt{k}$, and so~$kp \leq kx \leq kp + \sqrt{k}$. Therefore, for~$j \in J$, $|j-kx|\leq \sqrt{k}$, and by Observation~\ref{claim:proba_equal}, we get that~$\bbP\pa{ Y_1 = j } \geq \beta / \sqrt{k}$. Hence, we can rewrite Eq.~\eqref{eq:derivative_aux} as
        \begin{equation} \label{eq:derivative_aux2}
            \prob{k}{x+h}{>}{k}{p} - \prob{k}{x}{>}{k}{p} \geq \frac{\beta^2}{k} \sum_{j\in J} \bbP\pa{ H \geq 1 \mid Y_1 = j}.
        \end{equation}
        
        Now, let us find a lower bound on~$\bbP\pa{ H \geq 1 \mid Y_1 = j}$, for~$j \in J$. Note that, by definition, $Y_1=j$ if and only if~$|\{ i \text{ s.t. } X_i > x \}| = k-j$. Since~$X_i, 1 \leq i \leq k$, is uniformly distributed over~$[0,1]$,
        \begin{equation*}
            \bbP \pa{ x < X_i \leq x+h \mid X_i > x } = \frac{h}{1-x}.
        \end{equation*}
        Therefore, for every~$j\in J$
        \begin{equation*}
            \bbP \pa{H = 0 \mid Y_1=j} = \pa{1 - \frac{h}{1-x} }^{k-j} \leq \pa{ 1 - \frac{h}{1-x} }^{k - kp - \sqrt{k}}.
        \end{equation*}
        This implies that
        \begin{equation*}
            \sum_{j\in J} \bbP\pa{ H \geq 1 \mid Y_1 = j} \geq \sqrt{k} \cdot \pa{1-\pa{1 - \frac{h}{1-x} }^{k - kp - \sqrt{k}}}.
        \end{equation*}
        We have
        \begin{equation*}
             \lim_{h \rightarrow 0}~ \frac{1}{h} \cdot \sum_{j\in J} \bbP\pa{ H \geq 1 \mid Y_1 = j} \geq \lim_{h \rightarrow 0}~  \frac{\sqrt{k}}{h} \cdot \pa{1-\pa{1 - \frac{h}{1-x} }^{k - kp - \sqrt{k}}} = \frac{\sqrt{k} \pa{ k - kp - \sqrt{k} }}{1-x}.
        \end{equation*}
        Eventually, we get from Eq.~\eqref{eq:derivative_aux2}
        \begin{equation*}
            \frac{d}{dx} \prob{k}{x}{>}{k}{p} \geq \frac{\beta^2 (1-p)}{1-x} \cdot  \sqrt{k} + \smallO{k}{\infty} \pa{\sqrt{k}}.
        \end{equation*}
        We can conclude the proof of Claim~\ref{claim:derivative} for~$k$ large enough by taking, e.g., $\beta' = \frac{\beta^2 (1-p)}{2(1-x)}$.
    \end{proof}
    
    Now, we are ready to write the proof of Claim~\ref{claim:last_minute_obstacle}.
    
    \begin{proof} [Proof of Claim~\ref{claim:last_minute_obstacle}]
        We can rewrite Eq.~\eqref{eq:def_g} as 
        \begin{equation*}
            g(x,y) = \pa{y-\frac{1}{n}} \prob{\ell}{y}{\geq}{\ell}{x} + \pa{1-y} \cdot \prob{\ell}{y}{>}{\ell}{x} + \frac{1}{n}.
        \end{equation*}
        Hence,
        \begin{multline} \label{eq:deriv}
            \frac{d}{d y} g(x,y) \quad = \quad \bigg[ \prob{\ell}{y}{\geq}{\ell}{x} ~-~ \prob{\ell}{y}{>}{\ell}{x} \bigg] \\
            + \pa{y-\frac{1}{n}} \cdot \frac{d}{d y} \prob{\ell}{y}{\geq}{\ell}{x} \\
            + \pa{1-y} \cdot \frac{d}{d y}\prob{\ell}{y}{>}{\ell}{x}.
        \end{multline}
        The first term in Eq.~\eqref{eq:deriv} is equal to~$\bbP \pa{ B_\ell \pa{ y } = B_\ell \pa{ x } }$, which is positive. Moreover, $\bbP \pa{ B_\ell \pa{ y } \geq B_\ell \pa{ x } }$ is obviously increasing in~$y$, so the second term is also non-negative. By Claim~\ref{claim:derivative}, the third term in Eq.~\eqref{eq:deriv} satisfies
        \begin{equation*}
            \pa{1-y} \cdot \frac{d}{d y}\prob{\ell}{y}{>}{\ell}{x} \geq \pa{1-y} \cdot \beta' \cdot \sqrt{\ell} \geq \frac{\beta'}{4} \cdot \sqrt{\ell},
        \end{equation*}
        where the last inequality comes from the fact that~$x\in [1/3,2/3]$ and~$y\in [x,x+1/\sqrt{\ell}] \subseteq [1/4,3/4]$.
        For~$\ell$ large enough, this implies that
        \begin{equation*}
            \frac{d}{d y} g(x,y) \geq \frac{\beta'}{4} \cdot \sqrt{\ell} > 1,
        \end{equation*}
        which concludes the proof of Claim~\ref{claim:last_minute_obstacle}.
    \end{proof}



\section{Basic Observations regarding Algorithm FET} \label{app:evolution}
\subsection{Proof of Observation~\ref{lem:evolution}}
    Let~$I$ be the set of agents (including the source).
    Let~$I_t^{1} \subset I$ be the set of all {\em non-source} agents with opinion~$1$ at round~$t$.
    Recall that we condition on~$\x_t = \xx{t}$ and~$\x_{t+1} = \xx{t+1}$ (although we avoid writing this conditioning). In addition, the proof will proceed by conditioning on~$I_{t+1}^1= \II{t+1}{1}$. Since we shall show that the statements are true for every~$\II{t+1}{1}$, the lemma will hold without this latter conditioning.
    
    By definition of the protocol, and because it operates under the~$\pull$ model, ${\cnt}^{(i)}$ and~${\cntt}^{(i)}$ are obtained by sampling~$\ell$ agents uniformly at random in the population (with replacement) and counting how many have opinion~$1$. Therefore, conditioning on~$(\x_t,\x_{t+1})$ and~$I_{t+1}^1$, 
    \begin{itemize}
        \item[(i)] variables~$({\cnt_{t+1}}^{(i)})_{i \in I}$ and~$({\cntt_t}^{(i)})_{i \in I}$ are mutually independent, thus variables~$(Y_{t+2}^{(i)})_{i \in I}$ are mutually independent.
        \item[(ii)] for every~$i \in I$, ${\cnt_{t+1}}^{(i)} \sim \calB_\ell(\x_{t+1})$, and~${\cntt_{t}}^{(i)} \sim \calB_\ell(\x_{t})$, so we can write for every non-source agent~$i \in I_{t+1}^1$,
        \begin{equation*}
            \bbP \pa{ Y_{t+2}^{(i)} = 1 } = \prob{\ell}{\x_{t+1}}{\geq}{\ell}{\x_t},
        \end{equation*}
        and for every non-source agent~$i \notin I_{t+1}^1$,
        \begin{equation*}
            \bbP \pa{ Y_{t+2}^{(i)} = 1 } = \prob{\ell}{\x_{t+1}}{>}{\ell}{\x_t}.
        \end{equation*}
    \end{itemize}
    This establishes Eq.~\ref{eq:indiv_next_opinion}.
    Now, let us define independent binary random variables~$(X_j)_{1 \leq j \leq n}$, taking values in~$\{0,1\}$, as follows;
    \begin{itemize}
        \item $X_1 = 1$,
        \item for every~$j$ s.t.~$1 < j \leq n \cdot \x_{t+1}$, $\bbP\pa{X_j = 1} = \prob{\ell}{\x_{t+1}}{\geq}{\ell}{\x_t}$,
        \item for every~$j$ s.t.~$n \cdot \x_{t+1} < j \leq n$, $\bbP\pa{X_j = 1} = \prob{\ell}{\x_{t+1}}{>}{\ell}{\x_t}$.
    \end{itemize}
    We assumed the source agent to have opinion~$1$, so there are $n\x_t -1$ non-source agents with opinion~$1$ and~$n(1-\x_t)$ non-source agents with opinion~$0$.
    Therefore, by~(i) and~(ii) and by construction of the~$(X_j)_{1 \leq j \leq n}$, $\x_{t+2} = \frac{1}{n} \sum_{i \in I} Y_{t+2}^{(i)}$ is distributed as~$\frac{1}{n} \sum_{j=1}^n X_j$, which establishes the second statement in Observation~\ref{lem:evolution}.
    Computing the expectation (still conditioning on~$\x_t$, $\x_{t+1}$) is straightforward and does not depend on~$I_{t+1}^1$:
    \begin{align*}
        \bbE \pa{ \x_{t+2} } &= \pa{ \x_{t+1} - \frac{1}{n} } \cdot \prob{\ell}{\x_{t+1}}{\geq}{\ell}{\x_t} + (1-\x_{t+1}) \cdot \prob{\ell}{\x_{t+1}}{>}{\ell}{\x_t} + \frac{1}{n} \\
        &= \x_{t+1} \cdot \prob{\ell}{\x_{t+1}}{\geq}{\ell}{\x_t} + (1-\x_{t+1}) \cdot \prob{\ell}{\x_{t+1}}{>}{\ell}{\x_t} + \frac{1}{n}(1-\prob{\ell}{\x_{t+1}}{\geq}{\ell}{\x_t}) \\
        &= \prob{\ell}{\x_{t+1}}{>}{\ell}{\x_t} + \x_{t+1} \cdot \prob{\ell}{\x_{t+1}}{=}{\ell}{\x_t} + \frac{1}{n}(1-\prob{\ell}{\x_{t+1}}{\geq}{\ell}{\x_t}).
    \end{align*}
    This establishes Eq.~\eqref{eq:next_expectation}, and concludes the proof of Observation~\ref{lem:evolution}.
\qed

\begin{remark} \label{rem:evolution}
    From Observation~\ref{lem:evolution}, we obtain the following straightforward bounds: for every non-source agent~$i$,
    \begin{equation} \label{eq:remark1}
        \prob{\ell}{\x_{t+1}}{>}{\ell}{\x_t} ~~\leq~~ \bbP \pa{Y_{t+2}^{(i)} = 1} ~~\leq~~ \prob{\ell}{\x_{t+1}}{\leq}{\ell}{\x_t},
    \end{equation}
    and
    \begin{equation} \label{eq:remark2}
        \prob{\ell}{\x_{t+1}}{>}{\ell}{\x_t} - \frac{1}{n} ~~\leq~~ \bbE \pa{\x_{t+2}} ~~\leq~~ \prob{\ell}{\x_{t+1}}{\leq}{\ell}{\x_t} + \frac{1}{n}.
    \end{equation}
    Because of the source agent having opinion~$1$, the left hand side in Eq.~\eqref{eq:remark2} is loose (specifically, $-1/n$ is not necessary).
    Nevertheless, we will use this equation in the proofs, because it has a symmetric equivalent (w.r.t.~to the center of~$\calG$, $(\frac{1}{2},\frac{1}{2})$) which will allow our statements about~$\x_{t+2}$ to hold symmetrically for~$1-\x_{t+2}$, despite the asymmetry induced by the source.
\end{remark}

\begin{remark} \label{rem:evolution2}
    Eq.~\eqref{eq:next_expectation} in Observation~\ref{lem:evolution} implies the following convenient bounds: 
    \begin{gather} 
            \prob{\ell}{\x_{t+1}}{>}{\ell}{\x_t} + \x_{t+1} \cdot \prob{\ell}{\x_{t+1}}{=}{\ell}{\x_t} - \frac{1}{n} \nonumber \\
            < \quad \bbE \pa{ \x_{t+2} } \quad < \label{eq:expect_next_fraction_bounds} \\
            \prob{\ell}{\x_{t+1}}{>}{\ell}{\x_t} + \x_{t+1} \cdot \prob{\ell}{\x_{t+1}}{=}{\ell}{\x_t} + \frac{1}{n}. \nonumber
    \end{gather}
\end{remark}

\subsection{Affects of noise}\label{sec:noise}

When it comes to the central area, \yellow, we will need the following result to break ties.

\begin{lemma} \label{lem:anti_concentration}
    There exists a constant~$\beta > 0$ s.t. for~$n$ large enough, and if~$\bbE(\x_{t+2}) \in [1/3,2/3]$, then
    \begin{equation*}
        \bbP \pa{ \x_{t+2} \leq \bbE(\x_{t+2}) - 1/\sqrt{n} } , \bbP \pa{ \x_{t+2} \geq \bbE(\x_{t+2}) + 1/\sqrt{n} } \geq \beta.
    \end{equation*}
\end{lemma}
\begin{proof}
    Consider~$X_1,\ldots,X_n$ from the statement of Observation~\ref{lem:evolution}. We have (see the proof of Observation~\ref{lem:evolution})
    \begin{itemize}
        \item $X_1 = 1$,
        \item for every~$j$ s.t.~$1 < j \leq n \cdot \x_{t+1}$, $\bbP\pa{X_j = 1} = \prob{\ell}{\x_{t+1}}{\geq}{\ell}{\x_t}$,
        \item for every~$j$ s.t.~$n \cdot \x_{t+1} < j \leq n$, $\bbP\pa{X_j = 1} = \prob{\ell}{\x_{t+1}}{>}{\ell}{\x_t}$.
    \end{itemize}
    Let~$p = \prob{\ell}{\x_{t+1}}{\geq}{\ell}{\x_t}$ and~$q = \prob{\ell}{\x_{t+1}}{>}{\ell}{\x_t}$.
    By Observation~\ref{lem:evolution},
    \begin{equation*}
        \bbE(\x_{t+2}) = \bbE \pa{ \frac{1}{n} \sum_{i=1}^n X_i }
        = \x_{t+1} \cdot p + (1-\x_{t+1}) \cdot q + \frac{1}{n}\pa{1-p}.
    \end{equation*}
    By assumption on~$\bbE(\x_{t+2})$, this implies that
    \begin{equation*}
        \x_{t+1} \cdot p + (1-\x_{t+1}) \cdot q \in \left[ \frac{1}{3}-\frac{1}{n},\frac{2}{3} \right].
    \end{equation*}
    Moreover, we have that~$p-q = \bbP\pa{B_\ell(\x_{t+1} = B_\ell(\x_t)}$ which tends to~$0$ as~$n$ tends to infinity, i.e., $p$ and~$q$ are arbitrarily close.
    Hence, for~$n$ large enough, the last equation implies that~$p\in[1/4,3/4]$ and~$q\in[1/4,3/4]$.
    Let~$Y_p = \sum_{i=2}^{n\cdot \x_{t+1}} X_i$ and~$Y_q = \sum_{i=n\cdot\x_{t+1}+1}^n X_i$. These two variables are binomially distributed, and since~$p,q\in[1/4,3/4]$, there is a constant probability that~$Y_p \geq \bbE(Y_p)$, and there is a constant probability that~$Y_q \geq \bbE(Y_q)$ as well.
    Without loss of generality, we assume that~$\x_{t+1} \geq 1/2$ and focus on~$Y_p$ (if~$\x_{t+1} < 1/2$, then we could consider~$Y_q$ instead).
    Let~$m = n\cdot \x_{t+1}-1$ be the number of samples of~$Y_p$.
    In this case, $m \geq n/2-1$ tends to~$+\infty$ as~$n$ tends to~$+\infty$.
    Let~$\sigma_p = \sqrt{p(1-p)}$.
    By the central limit theorem (Theorem~\ref{thm:central_limit}), the random variable
    \begin{equation*}
        \frac{\sqrt{m}}{\sigma_p} \pa{\frac{1}{m}Y_p - p} = \frac{Y_p - \bbE(Y_p)}{\sigma_p \sqrt{m}}
    \end{equation*}
    converges in distribution to~$\calN(0,1)$.
    Moreover, $\var(Y_p) = m \sigma_p^2 = (n\cdot \x_{t+1}-1)p(1-p) \geq (n/2-1)p(1-p) \geq np(1-p)/3 $, so for any~$\epsilon>0$ and~$n$ large enough,
    \begin{align*}
        \bbP \pa{ Y_p \geq \bbE\pa{Y_p}+\sqrt{n} } &= \bbP\pa{ \frac{Y_p-\bbE\pa{Y_p}}{\sigma_p\sqrt{m}} \geq \frac{\sqrt{n}}{\sigma_p \sqrt{m}}} \\
        &\geq \bbP\pa{ \frac{Y_p-\bbE\pa{Y_p}}{\sigma_p \sqrt{m}} \geq \sqrt{\frac{3}{p(1-p)}}} \\
        &\geq 1 - \Phi\pa{\sqrt{\frac{3}{p(1-p)}}} - \epsilon.
    \end{align*}
    For~$\epsilon$ small enough, and because~$p$ is bounded, this probability is bounded away from zero. This concludes the proof of Lemma~\ref{lem:anti_concentration} (the other inequality can be obtained symmetrically).
\end{proof}

We can use the previous result to show that the Markov process $(\x_t,\x_{t+1})$ is sufficiently noisy so that it is never too likely to be at any given point~$(x,y)$.

\begin{lemma} \label{lem:noise}
    There is a constant~$c_1 = c_1(c) > 0$ (recall that the sample size is~$\ell = c \cdot \log n$), such that for any~$a \in [1/2-4\delta,1/2+4\delta]$, and any round~$t$ s.t.~$(\x_t,\x_{t+1}) \in \yellow'$,
    we have
    \begin{equation*}
        \bbP \pa{|\x_{t+2} - a| > \frac{1}{\sqrt{n}}} > c_1.
    \end{equation*}
\end{lemma}
\begin{proof}
    Follows directly from Lemma~\ref{lem:anti_concentration}.
\end{proof}

\section{Analyzing Domains}
\subsection{Green area} \label{app:green}

\begin{proof} [Proof of Lemma~\ref{lem:green}]
    Let us prove the first part and assume~$(\x_t,\x_{t+1}) \in \green{1}$ (the proof of the second part is analogous).
    By Eq.~\eqref{eq:remark1} in Remark~\ref{rem:evolution}, we have for every agent~$i$
    \begin{equation*}
        \bbP \pa{Y_i^{(t+2)} = 0} \leq \prob{\ell}{\x_{t+1}}{\leq}{\ell}{\x_t}.
    \end{equation*}
    By Lemma~\ref{lem:topdog}, we have
    \begin{equation*}
        \prob{\ell}{\x_{t+1}}{\leq}{\ell}{\x_t} \leq \exp \pa{- \frac{1}{2} \ell (\x_{t+1} - \x_t)^2} \leq \exp \pa{- \frac{1}{2} \ell \delta^2} =  \exp \pa{-\frac{c \delta^2}{2} \log n }.
    \end{equation*}
    Then, by the union bound,
    \begin{equation*}
        \bbP \pa{ \bigcup_{i\in I\setminus \{\text{source}\} } \pa{Y_i^{(t+2)} = 0} } \leq (n-1) \cdot \exp \pa{-\frac{c \delta^2}{2} \log n },
    \end{equation*}
    which tends to~$0$ as~$n$ goes to~$+\infty$, provided that~$c > 2/\delta^2$.
\end{proof}

\subsection{Purple area} \label{app:purple}

\begin{proof} [Proof of Lemma~\ref{lem:purple}]
    Let us prove the first part and assume~$(\x_t,\x_{t+1}) \in \purple{1}$ (the proof of the second part is analogous).
    By Eq.~\eqref{eq:remark2} in Remark~$\ref{rem:evolution}$,
    \begin{equation*}
        \bbE\pa{\x_{t+2}} \geq \prob{\ell}{\x_{t+1}}{>}{\ell}{\x_t} - \frac{1}{n}.
    \end{equation*}
    Since~$(\x_t,\x_{t+1}) \in \purple{1}$, and since in this area~$\x_{t+1}$ is at least a factor~$(1-\lambda_n)$ greater than~$\x_{t}$, we have
    \begin{equation*}
        \prob{\ell}{\x_{t+1}}{>}{\ell}{\x_t} \geq \prob{\ell}{(1-\lambda_n) \x_t}{>}{\ell}{\x_t}.
    \end{equation*}
    Let
    \begin{equation}\label{eq:sigma-purple}
        \sigma = \sqrt{\x_t(1-\x_t) + (1-\lambda_n) \x_t(1-(1-\lambda_n) \x_t)} > \sqrt{\x_t(1-\x_t)} > \sqrt{\frac{\x_t}{2}},
    \end{equation}
    where the last inequality is by the fact that $\x_t<1/2$ which follows from the definition of~$\purple{1}$. By Lemma~\ref{lem:underdog},
    \begin{equation*}
        \prob{\ell}{(1-\lambda_n) \x_t}{>}{\ell}{\x_t} > 1 - \Phi\pa{\frac{\sqrt{\ell} \lambda_n \x_t}{\sigma}} - \frac{C}{\sigma\sqrt{\ell}}.
    \end{equation*}
    We have (Eq.~\eqref{eq:sigma-purple} and definition of~$\purple{1}$)
    \begin{equation*}
        \sigma > \sqrt{\frac{\x_t}{2}} > \sqrt{\frac{1}{2\log n}}
    \end{equation*}
    so
    \begin{equation*}
        \frac{C}{\sigma\sqrt{\ell}} < \frac{\sqrt{2}C}{\sqrt{c}}.
    \end{equation*}
    If~$c$ is large enough (specifically, if~$c > 32C^2/\delta^2$), we obtain
    \begin{equation*}
        \prob{\ell}{(1-\lambda_n) \x_t}{>}{\ell}{\x_t} > 1 - \Phi\pa{\frac{\sqrt{\ell} \lambda_n \x_t}{\sigma}} - \frac{\delta}{4}.
    \end{equation*}
    We have
    \begin{equation*}
        0 \leq \frac{\sqrt{\ell} \lambda_n \x_t}{\sigma} \leq \sqrt{\ell} \lambda_n \sqrt{2\x_t} \leq \sqrt{\ell}\lambda_n = \frac{\sqrt{c}}{\log^\delta n} \tendsto{n}{+\infty} 0.
    \end{equation*}
    where the second inequality is by Eq.~\eqref{eq:sigma-purple}, and the third is because~$\x_t<1/2$. 
    So, for~$n$ large enough
    \begin{equation*}
        1 - \Phi\pa{\frac{\sqrt{\ell} \lambda_n}{\sigma}} - \frac{\delta}{4} > 1-\Phi(0) - \frac{\delta}{2} = \frac{1-\delta}{2}.
    \end{equation*}
    Overall, we have proved that if~$n$ is large enough, then $\bbE\pa{\x_{t+2}} > (1-\delta)/2$.
    By Observation~\ref{lem:evolution}, we can apply Chernoff's inequality (Theorem~\ref{thm:mult_chernoff_bound}) to get that~$\x_{t+2} > 1/2-\delta$ w.h.p. Since by definition of~$\purple{1}$ we have $1/2-\delta > \x_{t+1}+\delta$, we obtain $\x_{t+2} >\x_{t+1}+\delta$ w.h.p., which concludes the proof of the lemma.
\end{proof}

\subsection{Red area} \label{app:red}

\begin{proof}[Proof of Lemma~\ref{lem:red}]
    Without loss of generality, we assume that~$t_0 = 0$. We assume that~$(\x_{0},\x_{1}) \in \red{1}$ (the proof in the case that~$(\x_{0},\x_{1}) \in \red{0}$ is the same).
    First we note that for every round~$t$, by definition, if~$(\x_t,\x_{t+1}) \in \red{1}$ then~$\x_{t+1} < (1-\lambda_n) \x_t$. So, we can prove by induction on~$t$ that for every~$1\leq t \leq t_1$, \begin{equation*}
        \x_t < \x_{0} (1-\lambda_n)^t.
    \end{equation*}
    In particular, we have that~$\x_{t_1}<\x_{t_0}<1/2-3\delta$, and so~$(\x_{t_1},\x_{t_1+1}) \notin \yellow$ by definition of \yellow. 
    
    Also by definition, $\x_0 < 1/2$ and~$\x_t > 1/\log(n)$ for every~$0 \leq t \leq t_1$,  hence, we obtain from the last equation that
    \begin{equation*}
        \frac{1}{\log n} < \frac{1}{2} (1-\lambda_n)^t.
    \end{equation*}
    Taking the logarithm and rearranging, we get
    \begin{equation*}
         \log\pa{\frac{1}{2}} + \log(\log n) > t \cdot \log\pa{\frac{1}{1-\lambda_n}}.
    \end{equation*}
    We know that $\log(1-\lambda_n) < -\lambda_n$ and thus~$t \cdot \log\pa{1/(1-\lambda_n)} > t \lambda_n$. Together with the above equation, this gives
    \begin{equation*}
        t < \frac{1}{\lambda_n} \pa{\log\pa{\frac{1}{2}} + \log(\log n)} = o\pa{ \log^{1/2+2\delta} n },
    \end{equation*}
    which concludes the proof.
\end{proof}

\subsection{Cyan area} \label{app:cyan}

\begin{proof} [Proof of Claim~\ref{claim:cyan_small}]
    We note that, since~$\x_t < 1/\log(n)$, the probability that an agent does not see a~$1$ in round~$t$ is
    \begin{equation*}
        \bbP \pa{B_\ell\pa{\x_t} = 0} = (1-\x_t)^\ell > \pa{1-\frac{1}{\log n}}^\ell = \exp \pa{c \log(n) \log \pa{1-\frac{1}{\log n}}} > e^{-2c},
    \end{equation*}
    for~$n$ large enough. Moreover,
    \begin{equation*}
        (1-\x_{t+1})^\ell < 1 - \ell \x_{t+1} + \frac{1}{2}\ell^2\x_{t+1}^2,
    \end{equation*}
    so the probability that an agent  sees at least a~$1$ in round~$t+1$ is
    \begin{align*}
        \bbP(B_\ell(\x_{t+1})\geq 1) =
        1-(1-\x_{t+1})^\ell &> \ell \x_{t+1} \pa{1 - \frac{1}{2}\ell \x_{t+1}} > \frac{1}{2} \ell \x_{t+1},
    \end{align*}
    where the last inequality comes from the assumption that~$\x_{t+1} \leq 1/\ell$. 
    Eventually, we can write
    \begin{equation*}
        \prob{\ell}{\x_{t+1}}{>}{\ell}{\x_t} \geq \bbP \pa{B_\ell(\x_t) = 0} \cdot \bbP\pa{B_\ell(\x_{t+1}) \geq 1} \geq \frac{c}{2} \cdot e^{-2c} \cdot \x_{t+1} \log n = K \x_{t+1} \log n.
    \end{equation*}
    Hence, by Eq.~\eqref{eq:remark2} in Remark~\ref{rem:evolution}, $\bbE \pa{\x_{t+2}} \geq K \x_{t+1} \log n - 1/n$. By Observation~\ref{lem:evolution}, we can apply Chernoff's inequality (Theorem~\ref{thm:mult_chernoff_bound}) to conclude the proof of Claim~\ref{claim:cyan_small}.
\end{proof}

\begin{proof} [Proof of Claim~\ref{claim:cyan_intermediate}]
    The proof follows along similar lines as the proof of Claim~\ref{claim:cyan_small}.
    We note that, since~$\x_t < 1/\log(n)$, the probability that an agent does not see a~$1$ in round~$t$ is
    \begin{equation*}
        \bbP \pa{B_\ell\pa{\x_t} = 0} = (1-\x_t)^\ell > \pa{1-\frac{1}{\log n}}^\ell = \exp \pa{c \log n \log \pa{1-\frac{1}{\log n}}} > e^{-2c},
    \end{equation*}
    for~$n$ large enough. Moreover, the probability that an agent  sees at least a~$1$ in round~$t+1$ is
    \begin{equation*}
        \bbP(B_\ell(\x_{t+1})\geq 1) = 1-(1-\x_{t+1})^\ell \geq 1-\pa{1-\frac{1}{\ell}}^\ell > 1-\frac{1}{e}.
    \end{equation*}
    Eventually, we can write
    \begin{equation*}
        \prob{\ell}{\x_{t+1}}{>}{\ell}{\x_t} \geq \bbP\pa{B_\ell(\x_{t+1}) \geq 1} \cdot \bbP \pa{B_\ell(\x_t) = 0} \geq e^{-2c} \cdot \pa{1-\frac{1}{e}} = 2\gamma.
    \end{equation*}
    Hence, by Eq.~\eqref{eq:remark2} in Remark~\ref{rem:evolution},~$\bbE \pa{\x_{t+2}} \geq 2 \gamma - 1/n$. By Observation~\ref{lem:evolution}, we can apply Chernoff's inequality (Theorem~\ref{thm:mult_chernoff_bound}) to conclude the proof of Claim~\ref{claim:cyan_intermediate}.
\end{proof}
    
\begin{proof} [Proof of Claim~\ref{claim:cyan_large}]
    By assumption, $\x_{t+1}-\x_t \geq \gamma - 1/\log(n)$, and so by Lemma~\ref{lem:topdog}
    \begin{equation*}
        \prob{\ell}{\x_{t+1}}{>}{\ell}{\x_t} \geq 1 - \exp \pa{-\frac{1}{2} \ell \pa{\gamma-\frac{1}{\log n}}^2 } > \frac{3}{4}
    \end{equation*}
    for~$n$ large enough.
    Hence, by Eq.~\eqref{eq:remark2} in Remark~\ref{rem:evolution},~$\bbE \pa{\x_{t+2}} \geq 3/4 - 1/n$. By Observation~\ref{lem:evolution}, we can apply Chernoff's inequality (Theorem~\ref{thm:mult_chernoff_bound}) to conclude the proof of Claim~\ref{claim:cyan_large}.
\end{proof}

\subsection{Yellow area}

Recall the partitioning of the Yellow domain as illustrated in Figure \ref{fig:partition_yellow}. We analyze each of the resulting areas separately.  
\subsubsection{Area A} \label{app:yellow_A}

\begin{proof} [Proof of Lemma~\ref{lem:A2}]
    Without loss of generality, we assume that~$(\x_t,\x_{t+1}) \in \bfA_1$ (the same arguments apply to~$\bfA_0$ symmetrically).
    We have, provided that~$\delta$ is small enough and~$n$ is large enough,
    \begin{align*}
        \bbE(\x_{t+2}) &> \prob{\ell}{\x_{t+1}}{>}{\ell}{\x_t} + \x_{t+1} \cdot \prob{\ell}{\x_{t+1}}{=}{\ell}{\x_t} - \frac{1}{n} & \text{(by Remark~\ref{rem:evolution2})} \\
        &> \frac{1}{2} + 6  (\x_{t+1}-\x_t) + \pa{\x_{t+1}-\frac{1}{2}} \cdot \prob{\ell}{\x_{t+1}}{=}{\ell}{\x_t} & \text{(by Lemma~\ref{lem:handcrafted}, taking~$\lambda > 6$)} \\
        &> \frac{1}{2} + 6 (\x_{t+1}-\x_t). & \text{(by the definition of~$\bfA_1$)} 
    \end{align*}
    Hence,
    \begin{align*}
        \bbE(\x_{t+2}) - \x_{t+1} > \frac{1}{2}-\x_t + 5 (\x_{t+1}-\x_t) = (\x_{t+1} - (2\x_t - \frac{1}{2})) + 4 (\x_{t+1}-\x_t),
    \end{align*}
    and by definition of~$\bfA_1$, $(\x_{t+1} - (2\x_t - 1/2)) \geq 0$ and so
    \begin{equation} \label{eq:auxA2_expectation}
        \bbE(\x_{t+2}) >  4 (\x_{t+1}-\x_t) + \x_{t+1}.
    \end{equation}
    By Observation~\ref{lem:evolution}, we can apply Chernoff's inequality (Theorem~\ref{thm:mult_chernoff_bound}). Taking~$\epsilon = 2(\x_{t+1}-\x_t)/(4(\x_{t+1}-\x_t)+\x_{t+1})$,
    we have
    \begin{align*}
        \bbP \pa{ \x_{t+2} - \x_{t+1} \leq 2 (\x_{t+1}-\x_t) } &= \bbP \pa{ \x_{t+2} \leq (1-\epsilon)\pa{ 4(\x_{t+1}-\x_t) + \x_{t+1} } } \\
        &\leq \bbP\pa{ n\x_{t+2} \leq (1-\epsilon) \bbE(n\x_{t+2}) } & \text{(by Eq.~\eqref{eq:auxA2_expectation})} \\
        &\leq \exp \pa{-\frac{\epsilon^2}{2} \bbE(n\x_{t+2})} & \text{(by Theorem~\ref{thm:mult_chernoff_bound})} \\
        &\leq \exp \pa{-\frac{2\x_{t+1}}{\pa{4(\x_{t+1}-\x_t)+\x_{t+1}}^2} (\x_{t+1}-\x_t)^2n}. & \text{(by Eq.~\eqref{eq:auxA2_expectation} and definition of~$\epsilon$)} 
    \end{align*}
    Since~$\x_{t}$ and~$\x_{t+1}$ are close to~$1/2$, we have for~$\delta$ small enough
    \begin{equation} \label{eq:auxA2}
        \bbP \pa{ \x_{t+2} - \x_{t+1} > 2 (\x_{t+1} - \x_t) } \geq 1-\exp \pa{ - 3 (\x_{t+1}-\x_t)^2 n }.
    \end{equation}
    Now, we show that the event ``$\x_{t+2} - \x_{t+1} > 2 (\x_{t+1} - \x_t)$''  suffices for~$(\x_{t+1},\x_{t+2})$ to remain in~$\bfA_1$ or leave~$\yellow'$.
    \begin{claim} \label{claim:auxA2}
        If~$(\x_t,\x_{t+1}) \in \bfA_1$ and~$\x_{t+2} - \x_{t+1} > 2 (\x_{t+1} - \x_t)$, then~$(\x_{t+1},\x_{t+2})\in \bfA_1$ or~$(\x_{t+1},\x_{t+2})\notin \yellow'$.
    \end{claim}
    \begin{proof}
        If~$(\x_{t+1},\x_{t+2})\notin \yellow'$, the result holds. Otherwise, $(\x_{t+1},\x_{t+2})\in \yellow'$ and we have to prove that~$(\x_{t+1},\x_{t+2})$ satisfies~$\bfA_1$.(i) and~$\bfA_1$.(ii).
        First we prove that~$(\x_{t+1},\x_{t+2})$ satisfies~$\bfA_1$.(i):
        \begin{align*}
            \x_{t+2} &> \x_{t+1} + 2 (\x_{t+1} - \x_t) &\text{(by assumption in the claim)} \\
            & \geq \x_{t+1} &\text{(because~$(\x_t,\x_{t+1}) \in \bfA_1 \Rightarrow \x_{t+1} \geq \x_t$)} \\
            & \geq 1/2. & \text{(because~$(\x_t,\x_{t+1}) \in \bfA_1$ and by~$\bfA_1$.(i))}
        \end{align*}
        Then we prove that~$(\x_{t+1},\x_{t+2})$ satisfies~$\bfA_1$.(ii):
        \begin{align*}
            \x_{t+2} - \x_{t+1} &> 2(\x_{t+1} - \x_t) &\text{(by assumption in the claim)} \\
            &> (\x_{t+1} - \x_t) + (\x_t - 1/2) & \text{(because~$(\x_t,\x_{t+1}) \in \bfA_1$ and by~$\bfA_1$.(ii))} \\
            &= \x_{t+1} - 1/2,
        \end{align*}
        which concludes the proof of Claim~\ref{claim:auxA2}.
    \end{proof}
   Next, we apply Claim~\ref{claim:auxA2} to Eq.~\eqref{eq:auxA2} to establish~$(a)$.
    Eventually, $\x_{t+2} > \x_{t+1} + 4 (\x_{t+1}-\x_t) + 1/ \sqrt{n}$ implies $\x_{t+2} - \x_{t+1} >  + 2 (\x_{t+1}-\x_t)$ so we can use Claim~\ref{claim:auxA2},
    \begin{align*}
        &\bbP \pa{ (\x_{t+1},\x_{t+2}) \notin \yellow'\setminus\bfA_1 \medcap \x_{t+2} > \x_{t+1} + 4 (\x_{t+1}-\x_t) + 1/ \sqrt{n} } \\
        &= \bbP \pa{ \x_{t+2} > \x_{t+1} + 4 (\x_{t+1}-\x_t) + 1/ \sqrt{n} } & \text{(by Claim~\ref{claim:auxA2})} \\ 
        &> \bbP \pa{ \x_{t+2} > \bbE \pa{\x_{t+2}} + 1/ \sqrt{n} } & \text{(by Eq.~\eqref{eq:auxA2_expectation})} \\
        &> c_2 > 0, &
    \end{align*}
    where the existence of~$c_2$ is guaranteed by Lemma~\ref{lem:anti_concentration}. This establishes~$(b)$.
\end{proof}

\begin{proof} [Proof of Lemma~\ref{lem:A3}]
    Without loss of generality, we assume that~$(\x_t,\x_{t+1}) \in \bfA_1$ (the same arguments apply to~$\bfA_0$ symmetrically).
    Let us define event~$H_{t_0+1}$, that the system is either in~$A_1$ or out of~$\yellow'$ in round~$t_0+1$, and that the ``gap''~$(\x_{t_0+2}-\x_{t_0+1})$ is not too small. Formally,
    \begin{equation*}
        H_{t_0+1} :~~  (\x_{t_0+1},\x_{t_0+2}) \notin \yellow'\setminus\bfA_1 \medcap \x_{t_0+2} - \x_{t_0+1} > 1 / \sqrt{n}.
    \end{equation*}
    For~$t>t_0+1$, we define event~$H_t$, that the system is either in~$A_1$ or out of~$\yellow'$ in round~$t$, and that the gap~$(\x_{t+1}-\x_t)$ doubles. Formally,
    \begin{equation*}
        H_{t} :~~ (\x_{t},\x_{t+1}) \notin \yellow'\setminus\bfA_1 \medcap \x_{t+1} - \x_{t} > 2 (\x_{t} - \x_{t-1}).
    \end{equation*}
    We start with the following observation, which results directly from the definition of~$H_t$ for~$t \geq t_0+1$:
    \begin{equation} \label{eq:auxA3}
        \bigcap_{s=t_0+1}^{t-1} H_s \Rightarrow (\x_{t} - \x_{t-1}) > 2^{(t-t_0-2)} (\x_{t_0+2} - \x_{t_0+1}) \Rightarrow (\x_{t} - \x_{t-1}) > 2^{(t-t_0-2)} / \sqrt{n}.
    \end{equation}
    For every~$t > t_0+1$,
    \begin{align*}
        \bbP \pa{ H_{t} \left|~ \bigcap_{s=t0+1}^{t-1} H_s \right.} &> 1-\exp \pa{- 3 n \cdot (\x_{t} - \x_{t-1})^2 } & \text{(By Lemma~\ref{lem:A2})} \\
        &> 1-\exp \pa{- \frac{3}{4} \cdot 4^{(t-t_0-1)} }. & \text{(by Eq.~\eqref{eq:auxA3})} 
    \end{align*}
    By Lemma~\ref{lem:A2}~$(b)$, $(\x_{t_0+1},\x_{t_0+2}) \in \bfA_1$ and $\x_{t_0+2} - \x_{t_0+1} > 1 / \sqrt{n}$ w.p.~$c_2>0$. Together with the last equation and using the union bound, we get
    \begin{align*}
        \bbP \pa{ \bigcap_{t=t0+1}^{t_1} H_t } > c_2 \cdot \pa{ 1-\sum_{t=t0+2}^{t_1} \exp \pa{- \frac{3}{4} \cdot 4^{(t-t_0-1)} } }.
    \end{align*}
    We have the following very rough upper bounds
    \begin{align*}
        \sum_{t=t0+2}^{t_1} \exp \pa{- \frac{3}{4} \cdot 4^{(t-t_0-1)} } &< \sum_{t=t0+2}^{t_1} \exp \pa{- \frac{3}{4} \cdot 4 \cdot (t-t_0-1) } \\
        &< 2 \cdot e^{-3}.
    \end{align*}
    Hence, we have proved that for every~$t_1 > t_0+1$,
    \begin{equation*}
        \bbP \pa{ \bigcap_{t=t0+1}^{t_1} H_t } > c_2 \cdot \pa{ 1 - 2 \cdot e^{-3} } := c_3 > 0.
    \end{equation*}
    By Eq.~\eqref{eq:auxA3}, it implies that for every~$t_1 > t_0+1$,
    \begin{equation*}
        \bbP \pa{ (\x_{t_1} - \x_{t_1-1}) > 2^{(t_1-t_0-2)} / \sqrt{n} } > c_3.
    \end{equation*}
    For~$t_1$ large enough (e.g., $t_1 = t_0 + \log n$), this implies that~$(\x_{t_1-1},\x_{t_1}) \notin \yellow'$, otherwise the gap~$(\x_{t_1}-\x_{t_1-1})$ would be greater than~$8\delta$ which is the diameter of~$\yellow'$.
    This concludes the proof of Lemma~\ref{lem:A3}.
\end{proof}

\subsubsection{Area B: Proof of Lemma~\ref{lem:B2}} \label{sec:yellow_B}

    Without loss of generality, we assume that~$(\x_t,\x_{t+1}) \in \bfB_1$ (the same arguments apply to~$\bfB_0$ symmetrically).
    For any round~$t$, let~$H_t$ the event that~$(\x_t,\x_{t+1}) \in \bfB$ and~{\em (a)} of Lemma~\ref{lem:B1} holds.
    Let~$t_{\max} = t_0 + (\sqrt{c}/c_4) \cdot \log^{3/2} n$, and
    let~$X$ be the number of rounds between~$t_0$ and~$t_{\max}$ for which~$H_t$ does not happen. 
    Each time~$(a)$ in Lemma~\ref{lem:B1} doesn't hold, $(b)$ of Lemma~\ref{lem:B1} holds so there is a constant probability to leave~$\bfB$, so
    \begin{equation} \label{eq:aux_B2}
        \bbP \pa{ \text{for every~$t$ such that } t_0 \leq t \leq t_{\max}, (\x_{t},\x_{t+1}) \in \bfB \mid X = {\bf x} } \leq (1-c_5)^{\bf x}.
    \end{equation}
    Note that
    \begin{equation*}
        (1-c_5)^{(t_{\max}-t_0)/4} = \exp \pa{ \log(1-c_5) \cdot \frac{\sqrt{c}}{4c_4} \cdot \log^{3/2} n }.
    \end{equation*}
    This, together with Eq.~\eqref{eq:aux_B2}, implies that either (i) $X < (t_{\max}-t_0)/4$, or w.h.p. (ii) there is a time~$t_0 \leq t \leq t_{\max}$ such that~$(\x_{t},\x_{t+1}) \notin \bfB$ (in which case Lemma~\ref{lem:B2} holds).
    
    
    Now, consider case~(i). Let~$u_t = \x_t - 1/2$.
    By Lemma~\ref{lem:noise}, $u_{t_0} > 1/ \sqrt{n}$ with constant probability. Therefore, for the price of waiting up to an additional~$\log n$ rounds, we can assume that~$u_{t_0} > 1/ \sqrt{n}$ w.h.p. In what follows, we condition on that event.
    
   Note now, that whenever~$(\x_t,\x_{t+1}) \in \bfB_1$, by definition~$\x_{t+1}\geq \x_t > 1/2$ and so~$(\x_{t+1},\x_{t+2})$ cannot be in~$\bfB_0$. This implies that the system must remain in~$\bfB_1$ until it leaves~$\bfB$.
    Also by definition of~$\bfB_1$, if~$(\x_t,\x_{t+1}) \in \bfB_1$ then~$\x_t \leq \x_{t+1}$, and so~
    \begin{equation}\label{eq:u_t}
        u_t \leq u_{t+1}. 
        \end{equation}
    Moreover, by the fact that we are in case~(i), we have that the number of rounds~$t_0 \leq t \leq t_{\max}$ such that~$H_t$ happens is at least
    \begin{equation*}
        k := \frac{3(t_{\max}-t_0)}{4} = \frac{3\sqrt{c}}{4c_4} \cdot \log^{3/2} n.
    \end{equation*}
    Note that at each such round, by definition of~$H_t$ and~$(a)$ in Lemma~\ref{lem:B1},
    \begin{equation*}
        u_{t+1} > u_t \cdot \pa{1+\frac{c_4}{\sqrt{\ell}}}.
    \end{equation*}
    Hence, by Eq.~\eqref{eq:u_t},
    \begin{align*}
        u_{t_{\max}} > u_{t_0} \cdot \pa{1+\frac{c_4}{\sqrt{\ell}}}^k &= u_{t_0} \cdot \exp \pa{ k \log \pa{ 1+\frac{c_4}{\sqrt{\ell}} } } & \\
        &> u_{t_0} \cdot \exp \pa{ \frac{k \cdot c_4}{2\sqrt{\ell}} } & \text{(for~$n$ large enough)} \\
        &= u_{t_0} \cdot \exp \pa{ \frac{3}{8} \log n } & \text{(by definition of~$k$ and~$\ell$)} \\
        &= \frac{1}{2} \cdot u_{t_0} \cdot n^{4/3} & \\
        &> \frac{1}{2} n^{1/4}. & \text{(since we assumed~$u_{t_0} > 1/ \sqrt{n}$)}
    \end{align*}
    When~$n$ is large, this quantity is larger than~$1$, hence~(i) is impossible unless~$(\x_t,\x_{t+1}) \notin \bfB$. This concludes the proof of Lemma~\ref{lem:B2}.
\qed

\subsubsection{Area C: Proof of Lemma~\ref{lem:C}} \label{sec:yellow_C}

    Without loss of generality, we assume that~$(\x_t,\x_{t+1}) \in \bfC_1$ (the same arguments apply to~$\bfC_0$ symmetrically).
    By Observation~\ref{lem:evolution},
    we have
    \begin{equation*}
        \bbE(\x_{t+2}) = \prob{\ell}{\x_{t+1}}{>}{\ell}{\x_t} + \x_{t+1} \cdot \prob{\ell}{\x_{t+1}}{=}{\ell}{\x_t} - \frac{1}{n}.
    \end{equation*}
    By Lemma~\ref{lem:handcrafted} (taking~$\lambda > 2$), this becomes
    \begin{equation} \label{eq:auxC}
         \bbE(\x_{t+2}) > \frac{1}{2} + 2 \cdot (\x_{t+1}-\x_t) - \pa{\frac{1}{2}-\x_{t+1}} \cdot \prob{\ell}{\x_{t+1}}{=}{\ell}{\x_t}.
    \end{equation}
    
    {\em Case 1.} If~$(\x_{t+1}-\x_t) > 1/2 - \x_{t+1}$, then Eq.~\eqref{eq:auxC} implies
    \begin{align*}
        \bbE(\x_{t+2}) &> \frac{1}{2} + 2 \cdot (\x_{t+1}-\x_t) - \pa{\frac{1}{2}-\x_{t+1}} > \frac{1}{2} + (\x_{t+1}-\x_t) > \frac{1}{2},
    \end{align*}
    so with constant probability~$\x_{t+2} > 1/2$ and thus~$(\x_{t+1},\x_{t+2}) \in \bfA_1$ or is not in~$\yellow'$. \\
    
    {\em Case 2.} Else, if~$(\x_{t+1}-\x_t) \leq 1/2 - \x_{t+1}$, Eq.~\eqref{eq:auxC} rewrites
    \begin{align*}
        \bbE(\x_{t+2}) &> \frac{1}{2} \pa{\frac{1}{2} + \x_{t+1}} + \frac{1}{2} \pa{ \frac{1}{2} + 4 \cdot (\x_{t+1}-\x_t) - 2 \pa{\frac{1}{2}-\x_{t+1}} \cdot \prob{\ell}{\x_{t+1}}{=}{\ell}{\x_t} - \x_{t+1}} \\
        &= \frac{1}{2} \pa{\frac{1}{2} + \x_{t+1}} + \frac{1}{2} \pa{ 4 \cdot (\x_{t+1}-\x_t) + \pa{\frac{1}{2}-\x_{t+1}} \pa{ 1 - 2 \cdot \prob{\ell}{\x_{t+1}}{=}{\ell}{\x_t} }  }
    \end{align*}
    Since~$(\x_t,\x_{t+1}) \in \bfC_1$, we have~$\x_{t+1} \geq \x_t$ and $1/2 > \x_{t+1}$. Moreover, for~$\ell$ large enough, $1 - 2 \cdot \bbP \pa{B_\ell(\x_{t+1})  = B_\ell(\x_{t})} > 0$. Hence,
    \begin{equation*}
        \bbE(\x_{t+2}) > \frac{1}{2} \pa{\frac{1}{2} + \x_{t+1}},
    \end{equation*}
    so, by Lemma~\ref{lem:anti_concentration}, with constant probability~$\x_{t+2} > \pa{1/2 + \x_{t+1}} / 2$, i.e., $\x_{t+2} - \x_{t+1} > 1/2 - \x_{t+2}$. If so, Case 1 applies and with constant probability, $(\x_{t+2},\x_{t+3}) \in \bfA_1$ or is not in~$\yellow'$.
This concludes the proof of Lemma~\ref{lem:C}.
\qed


\end{document}